\documentclass{aptpub}

\usepackage[T1]{fontenc}
\usepackage{amssymb, amsfonts}
\usepackage[dvips]{graphicx}
\usepackage{dsfont}				
\usepackage{bbm}
\usepackage{graphicx}
\usepackage{wrapfig}
\usepackage[percent]{overpic}
\usepackage{pdfpages}
\usepackage{enumerate}
\usepackage{tikz}

\AtBeginDocument{}   
\newcounter{saveeqn}                                 

\newcommand{\Lu}{\ensuremath{{\overline{D}}}}
\newcommand{\M}{\ensuremath{{\mathcal M}}}

\newcommand{\R}{\ensuremath{{\mathds R}}}

\newcommand{\leer}{\vspace{0.5cm}}					

\newtheorem{satz}[thm]{Theorem}

\newtheorem{bem}[thm]{Remark}

\newtheorem{defi}[thm]{Definition}

\DeclareMathOperator*{\argmax}{arg\,max}

\allowdisplaybreaks

\usepackage[T1]{fontenc}
\usepackage{lmodern}
\usepackage{epstopdf}

\authornames{S\"oren Christensen, Albrecht Irle, and Andreas Ludwig} 
\shorttitle{Optimal portfolio selection under vanishing fixed costs} 

\begin{document}

\title{Optimal portfolio selection under vanishing fixed transaction costs} 

\authorone[Hamburg University]{S\"oren Christensen} 
\authortwo[Christian-Albrechts-University Kiel]{Albrecht Irle} 
\authorthree[Christian-Albrechts-University Kiel]{Andreas Ludwig} 
\addressone{Department of Mathematics, SPST, Bundesstra\ss e 55, 20146 Hamburg, Germany} 
\addresstwo{Mathematisches Seminar,
	Ludewig-Meyn-Str. 4, D-24098 Kiel, Germany} 
\addressthree{Mathematisches Seminar,
	Ludewig-Meyn-Str. 4, D-24098 Kiel, Germany} 

\begin{abstract}
In this paper, asymptotic results in a long-term growth rate portfolio optimization model under both fixed and proportional transaction costs are obtained. More precisely,
the convergence of the model when the fixed costs tend to zero is investigated. A suitable limit model with purely proportional costs
is introduced and
the convergence of optimal boundaries, asymptotic growth rates, and optimal risky fraction processes is rigorously proved. The results are based on an in-depth analysis of the convergence of the solutions to the corresponding HJB-equations. \end{abstract}

\keywords{impulse control; Hunt processes; threshold rules; excessive functions; threshold rules} 

\ams{91G10}{93E20} 

\section{Introduction}

In this article, we study portfolio optimization problems in a Black-Scholes market using the 
\textit{Kelly-criterion} of maximizing the asymptotic growth rate 
$$\liminf\limits_{T\to\infty}\frac{1}{T}E(\log V_T),$$
going back to Kelly in \cite{K56}. Without transaction costs, the problem can be solved along the lines of \cite{M69} 
and it can essentially be obtained that the optimal strategy consists of keeping the fraction of total wealth invested in each asset constant. In Merton's honor this fraction is called the \textit{Merton ratio} or \textit{Merton fraction}.

It is, however, more realistic that an investor faces different types of costs such as brokerage and management fees, search and information costs, commissions and many others. For these kind of costs the notion \textit{transaction costs} shall be used. There are basically three different approaches to model transaction costs. The earliest approach goes back to Magill and Constantinides in \cite{MC76} and considers \textit{proportional transaction costs}, i.e.$,~$the investor has to pay a fixed proportion of each trading volume.
For the discounted consumption criterion they conjectured for the case of one stock that the optimal strategy is to keep the risky fraction process in a certain interval $[A,B]\subset\,]0,1[$ with minimal effort. The resulting process is a reflected diffusion with infinitesimal trading at the boundaries $A,B$. This was then proved by Davis and Norman for the logarithmic and power utility in \cite{DN90}. Shreve and Soner analyzed rigorously the optimal strategy in \cite{SS94} and established the value function as the unique solution to the Hamilton-Jacobi-Bellman equation, using viscosity solution techniques. The general case $d\geq 1$ was then treated in \cite{AMS96}.\\
The problem of maximizing the asymptotic growth rate under proportional transaction costs was solved by Taksar, Klass, and Assaf in \cite{TKA88} for one stock, where only selling stocks was punished. They derived the same structure of the optimal strategy as Davis and Norman. Using viscosity techniques Akian, Sulem, and Taksar proved existence of a solution to the HJB-equation, which is of variational form, in the $d$-stock case in \cite{AST01}.

Considering \textit{fixed transaction costs}, i.e.$~$at every transaction a fixed proportion of the investor's wealth has to be paid, solves the problem of the rather unfeasible infinitesimal trading of purely proportional costs and represents a second approach.
Here, the so-called \textit{impulse control theory} is applied and the strategies are of impulsive form, i.e.$~$they are given via a sequence $(\tau_n,\eta_n)_{n\in\N_0}$ consisting of stopping times $\tau_n$ with respect to $(\mathcal{F}_t)_{t\geq 0}$ that denote the trading times and satisfy
\begin{enumerate}[{(}i{)}]
	\item $0=\tau_0\leq\tau_1\leq\tau_2\leq\ldots$ with $\tau_n<\tau_{n+1}$ on $\{\tau_n<\infty\}$ for all $n\in\N_0$,
	\item $P\left(\lim\limits_{n\to\infty}\tau_n=\infty\right)=1$
\end{enumerate}
and $\mathcal{F}_{\tau_n}$-measurable $\R^d$-valued random variables $\eta_n$, describing the trading volume at $\tau_n$. The optimal strategy in the one stock case is to wait until the risky fraction process reaches the boundary of some interval $[A,B]\subset \,]0,1[$ containing the Merton fraction and then trade back to some fixed fraction in $]A,B[$ near the Merton fraction and restart the process. This optimal behavior was described for the Kelly criterion in \cite{MoPl95}. Bielecki and Pliska then generalized these results 
in several ways by characterizing the optimal strategies in terms of solutions to quasi-variational inequalities in \cite{BP00}, while existence and uniqueness results for solutions to these HJB-equations in quasi-variational form were established by Nagai in \cite{N04} by applying a coordinate transformation to avoid degeneracy.

Despite of the feasibility of the optimal trading strategies under fixed transaction costs, the cost structure seems rather unrealistic from the practitioner's point of view. To overcome this problem a combination of fixed and proportional transaction costs was suggested. In some cases, as in \cite{EH88}, \cite{K98}, and \cite{belak2016utility}, the fixed component of the transaction costs is a constant amount not depending on the wealth. Here, the authors derive solutions for the discounted consumption criterion for the linear utility, asymptotically for the exponential utility, and existence of optimal strategies, resp. Asymptotic results for vanishing fixed costs
were recently obtained in \cite{AltaroviciEtAl2015}, \cite{F16}, \cite{melnyk2017small}, where in particular the last paper considers a generalization of our setting.
These results are, however, different in nature to the results obtained in this article as the authors do not show convergence of the optimal strategies, but construct asymptotically optimal strategies, which are obviously suboptimal for all fixed positive costs. We also want to mention \cite{guasoni2015nonlinear} where a market with price-impact proportional to a power of the order flow is considered and asymptotically explicit formulas are obtained. The precise results and techniques, however, are quite different to ours.

Our attention in the present work will be focused on the cases where the fixed component of the transaction costs is a fixed proportion of the investor's wealth, as described above, under the maximization of the asymptotic growth rate.
The trading strategies are therefore of impulsive form and the costs paid by the investor at time $\tau_n$ are of the size
\begin{equation*}
c(V_{\tau_n},\eta_n)=\delta V_{\tau_n}+\gamma|\eta_n|,
\end{equation*}
where $\delta\in\,]0,1[$ denotes the fixed part of the costs, while the proportional part is described by $\gamma{\geq 0}$ with $\gamma<1-\delta$.\\
Inspired by the results from \cite{MoPl95} for purely fixed costs, Irle and Sass introduced in \cite{IS05} the class of the so-called \textit{constant boundary strategies} for the one-dimensional case and proved their optimality via a solution to the HJB-equation in quasi-variational form, rigorously constructed in \cite{IS06}. These strategies can be described via four constants $a<\alpha\leq\beta<b\in\,]0,1[$ such that $[a,b]$ is the continuation region, or the no-trade region, of the risky fraction process that is then restarted via trading in $\alpha$, respectively in $\beta$, when it reaches the boundary in $a$, respectively in $b$.\\
These optimal constants may easily be computed numerically, of course depending on the parameters of the Black-Scholes market and on the transaction costs, see  \cite{IS05},  \cite{IS06}, so these results provide a semi-explicit solution to the problem with fixed and proportional transaction costs.

For $d\geq 2$ stocks it seems very difficult to obtain results on the geometric structure of the optimal strategy, but results on the existence of an optimal strategy as solutions to the corresponding HJB-equation were obtained by Tamura. He adapted methods from \cite{N04} to the case $d=1$ in \cite{T06} and to the general case $d\geq 1$ in \cite{T08}, and derived an optimal strategy via a solution of the HJB-equation in quasi-variational form. This solution is obtained by perturbation methods and results on quasi-variational inequalities with discount factors from \cite{BL82} and \cite{BL84}, so does not provide any ready insight on the geometric structure of the no-trade region and on the way, how trading back should be done. A rough guess would see optimal strategies as given by two surfaces in $d$-dimensional space: As soon as the outer surface is reached by the risky fraction process, trading back occurs to some point at the inner surface. The performance of such strategies was investigated in \cite{IP09} but the possible optimality of such strategies was not considered here.

In this paper, we take as a starting point the model of \cite{IS06} with transaction costs $c(\nu,\eta) = \delta\nu + \gamma|\eta|$. As described above, an optimal strategy is given by the levels $a<\alpha\leq \beta<b$:

\begin{figure}[ht]
	\begin{center}
		\begin{tikzpicture}[scale=1.5]
		\draw (-3,1) -- (3,1);
		\draw (-3.3,1) node {$b$};
		\draw (-3,0.7) -- (3,0.7);
		\draw (-3.3,.7) node {$\beta$};
		\draw (-3,0) -- (3,0);
		\draw (-3.3,0) node {$\alpha$};
		\draw (-3,-0.3) -- (3,-0.3);
		\draw (-3.3,-0.3) node {$a$};
		\draw[->] (-2,1) -- (-2,0.72);
		\draw[->] (-1,1) -- (-1,0.72);
		\draw[->] (2,-.3) -- (2,-0.03);
		\draw[->] (1,-.3) -- (1,-0.03);
		\end{tikzpicture}
	\end{center}
\end{figure}

Trading back from $b$ to $\beta$ or $a$ to $\alpha$, resp., causes fixed costs and proportional costs. For $\delta$ becoming smaller, the punishment for frequent trading with small volume gets less, so we should expect the differences $b-\beta$ and $\alpha-a$ in the optimal boundaries to become small. The optimally controlled risky fraction process will have an increasing number of small trades from $b$ to $\beta$ and $a$ to $\alpha$ resp., so that in the limit, as $\delta$ tends to zero for fixed $\gamma$, a diffusion with reflecting boundaries should turn up.

In this article we shall give the precise mathematical statements and proofs for these heuristics thus obtaining a connection between the different transaction cost models. We show the convergence of the model, when the fixed costs $\delta$ tend to zero, to a model with only proportional costs corresponding to that from \cite{TKA88}, by proving the convergence of the optimal boundaries, asymptotic growth rates, and optimal risky fraction process. The obtained result can also be of interest when only considering the limiting model with pure proportional costs. As the optimal strategies are of reflection type, a discretization is needed to make them realizable. Then, constant boundary strategies with upper- and lower boundaries close to each other are plausible candidates and the results of this paper provide a rigorous framework to justify this. We want to remark here that the -- on a first view -- artificial fixed costs can be interpreted as opportunity costs for using discretized reflection strategies, see \cite{christensen2012optimal}.

The article is organized in the following manner. For the sake of completeness, we introduce in Section \ref{Portfolio} the model with fixed and proportional costs from \cite{IS06} and combine the results derived in \cite{IP09} related to the representation of the maximization problem via risky fraction processes and the subsequent application of coordinate transformation to the whole space $\R$. 

In Section \ref{280911i}, we introduce a suitable model with only proportional costs similar to that of \cite{TKA88} and state a verification theorem for the corresponding HJB-approach.

The main results can then be found in Section \ref{nofix}. We treat the convergence in case of vanishing fixed costs $\delta$
and establish that the optimal boundaries $$a=a(\delta)<\alpha=\alpha(\delta)\leq\beta=\beta(\delta)<b=b(\delta)\in\,]0,1[$$
converge to limits
$$\lim\limits_{\delta\to 0}a(\delta)=A=\lim\limits_{\delta\to 0}\alpha(\delta),~\lim\limits_{\delta\to 0}\beta(\delta)=B=\lim\limits_{\delta\to 0}b(\delta)$$
for vanishing fixed costs $\delta$. Furthermore, the values are proved to converge and 
we obtain the optimality of the reflected risky fraction process on $[A,B]$ for the limit model with $\delta=0$ introduced in Section \ref{280911i}. As a byproduct, this yields the uniqueness of the optimal boundaries $A,B$ in the problem with pure proportional transaction costs. Furthermore, the weak convergence of the risky fraction processes to the optimal reflected risky process is inferred.
These results are derived from a careful analysis of the convergence of the solutions to the corresponding HJB-equations. It should be noted that the way to prove convergence of the optimal impulse strategies to a singular control strategy for vanishing fixed transaction costs is rather general in nature and could be applied to other classes of ergodic problems, too. However, the one-dimensional nature of the underlying processes is used in our approach, so that it is not straightforward to, e.g., treat problems with constant costs (not depending on the wealth) with this technique. 


Looking at the transaction costs $c(\nu,\eta) = \delta\nu + \gamma|\eta|$, the second natural problem is to look at the convergence as $\gamma$ tends to zero for fixed $\delta$. For $\gamma$ becoming smaller, trading with large volumes inflicts less costs, so that we expect fewer trades with larger volumes. 
It is an immediate conjecture that, as $\gamma$ tends to zero for fixed $\delta$, the optimal strategy tends to the optimal strategy in the model with $\gamma$ as treated in \cite{MoPl95}, where $\alpha=\beta$. That this conjecture is true may be shown by methods similar to those in this article but the precise proofs turn out to differ substantially from the arguments in this paper. In order not to overburden the contents here, we refer this to a further paper; the key concepts may be found in \cite{L12}. That thesis also provides more details on the results in Sections \ref{Portfolio} and \ref{280911i}.


\section{Portfolio model with fixed and proportional transaction costs}\label{Portfolio}
\subsection{Description of the model and preliminary results}
In this subsection we introduce the portfolio model with fixed and proportional transaction costs from \cite{IS05}, \cite{IS06}, and \cite{IP09}. 
We consider a financial market model with one bond $B$ and one stock $S$ satisfying
\begin{align}\label{030910a}
\begin{split}
&dB_t=rB_tdt,~~ B_0=b_0,\;\;
dS_t=\mu S_tdt+\sigma S_tdW_t,~~ S_0=s_0
\end{split}
\end{align}
with constant starting values $s_0, ~b_0$ and a Brownian motion $W$ adapted to the standard filtration $(\mathcal{F}_t)_{t\geq 0}$ on a probability space $(\Omega,\mathcal{F},P,(\mathcal{F}_t)_{t\geq 0})$. $r\geq0$ denotes the interest rate, $\mu\in\R$ denotes the trend and $\sigma>0$ is the volatility.
By $X_t$, respectively $Y_t$, we denote the amount of money the investor has invested in the bond, respectively the stock, at time $t$ and define $V_t:=X_t+Y_t$ to be the \textit{wealth} or \textit{portfolio value}.
We do not allow short selling or borrowing and therefore have
\begin{equation}\label{091211a}
X_t\geq 0\mbox{ and }Y_t\geq 0.
\end{equation}
Assuming $V_t>0$, we can define the \textit{risky fraction process} $h_t$ by
\begin{equation}
h_t:=\frac{Y_t}{V_t}\mbox{ and }h_t^0:=\frac{X_t}{V_t}\mbox{ for every }t\geq 0.
\end{equation}
The assumption $V_t>0$ above is not really restrictive, since for all trading strategies considered in the following it will be a consequence of \eqref{091211a}.
In our future models the investor will face fixed and proportional transaction costs, therefore sensible trading can only occur at discrete times $\tau_n,\,n\in\N_0$, and we denote by $\eta_n$ the transaction volume in the stock at time $\tau_n,\,n\in\N_0$. Hence the natural class of trading strategies are \textit{impulse control strategies} $\tau_n,\,n\in\N_0$ satisfying
\begin{equation}\label{151210a}
\eta_n=0 \mbox{ on }\{\tau_n=\infty\}\mbox{ for all }n\in\N.
\end{equation}

Here, the condition $\tau_0=0$ introduced in the definition of impulse control strategies is just for technical reasons and since trading in $\infty$ will have no effect on the growth rate in \eqref{020910f}, we restrict ourselves to \eqref{151210a} for simplicity.
We will later have to define the class of \textit{admissible} trading strategies
\begin{equation}
\mathcal{A}:=\{K\,:\,K\mbox{ is an admissible impulse control strategy}\},
\end{equation}
which of course will depend on the transaction costs and how they are paid. For every $K\in\mathcal{A}$ we consider the corresponding wealth process $(V_t^K)_{t\geq 0}$ and the expected growth rate
\begin{equation}\label{020910f}
J(K):=\liminf\limits_{T\to\infty}\frac{1}{T}E\bigl(\log V_T^K\bigr).
\end{equation}
The general aim is to maximize $J$ over all trading strategies, i.e.$~$to find the value
\begin{equation}\label{271011a}
\rho:=\sup\limits_{K\in\mathcal{A}}J(K)
\end{equation}
and the corresponding maximizing trading strategy $K^*$, if existing.
We assume the investor faces investment fees, given by the cost function
\begin{equation}\label{cost}
c:[0,\infty[\times \R\rightarrow[0,\infty[,~(x,\eta)\mapsto \delta x+\gamma|\eta|,
\end{equation}
where $\delta\in\,]0,1[$ denotes the fraction of the portfolio value (fixed costs) and $\gamma\in[0,1-\delta[$ are the fractions of the transaction volume $\eta$ (proportional costs).
For impulse control strategies $(\tau_n,\eta_n)_{n\in\N_0}$ 
we define that after the $n$-th trading the assets become
\begin{align}\label{030910b}
\begin{split}
X_{\tau_n}&=X_{\tau_n-}
-\eta_n-c(V_{\tau_n-},\eta_n),\;\;
Y_{\tau_n}=Y_{\tau_n-}+\eta_n,
\end{split}
\end{align}
on $\{\tau_n<\infty\}$ and hence
\begin{align}\label{120111b}
\begin{split}
V_{\tau_n}&=V_{\tau_n-}-c(V_{\tau_n-},\eta_n),\;\;
h_{\tau_n}=V_{\tau_n}^{-1}Y_{\tau_n}.
\end{split}
\end{align}
Between the trading times the processes are supposed to evolve according to \eqref{030910a}, i.e.$~$the number of bonds or stocks held by the investor has to be constant. Since we do not allow short selling or borrowing we can now define admissible trading strategies.

\begin{defi}\label{zulmonStrat}
	An impulse control strategy $K=(\tau_n,\eta_n)_{n\in\N_0}$ is an \textit{admissible monetary strategy} if the
	corresponding processes from \eqref{030910b} satisfy
	\begin{align}\label{030910c}
	\begin{split}
	X_{\tau_n},Y_{\tau_n}&\geq 0~ \mbox{ on }\{\tau_n<\infty\}\mbox{ for all }n\in\N_0.
	\end{split}
	\end{align}
	Since admissibility clearly depends on the starting values $v_0,h_0$, we define
	\begin{equation}
	\mathcal{A}_{v_0,h_0}:=\{K\,:\,K\mbox{ is an admissible monetary strategy for }v_0,h_0\}
	\end{equation}
	and we will sometimes write $X^{K,v_0,h_0},Y^{K,v_0,h_0}$ and $V^{K,v_0,h_0}$ for the processes if needed.
\end{defi}


%

Given an admissible monetary strategy $(\tau_n,\eta_n)_{n\in\N_0}$ and the corresponding risky fraction process $(h_t)_{t\geq 0}$ we get
\begin{equation}\label{151210b}
\xi_n:=h_{\tau_n}\in[0,1]~\mbox{ on }\{\tau_n<\infty\}\mbox{ for all }n\in\N_0.
\end{equation}

\begin{defi}\label{zulpropStrat}
	An impulse control strategy $\widetilde{K}=(\tau_n,\xi_n)_{n\in\N_0}$ is an \textit{admissible proportional strategy} if $\xi_n\in[0,1]$ on $\{\tau_n<\infty\}$ for all $n\in\N_0$.
	In analogy to Definition \ref{zulmonStrat} we define
	\begin{equation}
	\widetilde{\mathcal{A}}:=\left\{\widetilde{K}\,:\,\widetilde{K}\mbox{ is an admissible proportional strategy}\right\}.
	\end{equation}
\end{defi}

\begin{bem}\label{030910f}
	If we define $\xi_n\equiv 0$ on $\{\tau_n=\infty\}$ in \eqref{151210b}, we can see that an admissible proportional strategy can by deduced from an admissible monetary strategy. In fact, there is a one-to-one correspondence between these two kinds of strategies, see Lemma 2.3 and Theorem 2.4 in \cite{IP09}.
\end{bem}

%

We can reformulate the cost function associated to risky fractions as follows, see \cite{IS06}.
\begin{prop}\label{120111}
	There exists a constant $\kappa\in\,]0,1[$ and a Lipschitz continuous function
	\begin{equation}\label{120111d}	C:[0,1]^2\rightarrow[0,\kappa],
	\end{equation}
	such that for every admissible proportional strategy $\widetilde{K}=(\tau_n,\xi_n)_{n\in\N_0}$
	\begin{equation}\label{120111c}	V_{\tau_n}=(1-\delta)(1-C(h_{\tau_n-},\xi_n))V_{\tau_n-}
	\end{equation}
	holds on $\{\tau_n<\infty\}$ for every $n\in\N_0$. More, explicitly,
	\begin{equation}\label{150211b}
	\widehat{C}(h,\xi):=(1-\delta)(1-C(h,\xi))=\begin{cases}
	\frac{1-\delta+\gamma h}{1+\gamma \xi},&\;\xi\geq \frac{h}{1-\delta},\\
	\frac{1-\delta-\gamma h}{1-\gamma \xi},&\;\xi< \frac{h}{1-\delta}.
	\end{cases}
	\end{equation}
\end{prop}

As described before, the main objective is to find the optimal growth rate
\begin{equation}\label{eq:rho}
\rho_{v_0,h_0}=\sup\limits_{K\in\mathcal{A}_{v_0,h_0}}J_{v_0,h_0}(K)
=\sup\limits_{K\in\mathcal{A}_{v_0,h_0}}
\liminf\limits_{T\to\infty}\frac{1}{T}E\left(\log V_T^{K,v_0,h_0}\right),
\end{equation}
which we later prove to be independent of $v_0,h_0$, and a corresponding maximizing monetary strategy $K^*_{v_0,h_0}$, if existing, satisfying
\begin{equation}
\rho_{v_0,h_0}=J_{v_0,h_0}\left(K^*_{v_0,h_0}\right).
\end{equation}
Remark \ref{030910f} shows that, by setting $\tilde{J}_{v_0,h_0}(\widetilde{K}):=J_{v_0,h_0}(K)$, we can do this by studying
\begin{equation}
\rho_{v_0,h_0}=\sup\limits_{\widetilde{K}\in\widetilde{\mathcal{A}}}
\tilde{J}_{v_0,h_0}\bigl(\widetilde{K}\bigr)
=\sup\limits_{\widetilde{K}\in\widetilde{\mathcal{A}}}
\liminf\limits_{T\to\infty}\frac{1}{T}E\left(\log V_T^{\widetilde{K},v_0,h_0}\right).
\end{equation}

\subsection{A transformation of the problem}\label{reform}
In \cite{T08} a coordinate transformation $\psi$, already introduced in \cite{N04} for a model with only fixed transaction costs, is applied to the risky fraction process $h$ in order to avoid degeneracy at the boundary of the state space. The resulting diffusion $\psi(h)$ is then of a much easier structure and the transformed problem can be solved via the Hamilton-Jacobi-Bellman approach. 
We introduce the bijective transformations
\begin{equation}\label{psi}
\psi:]0,1[\rightarrow\R,~h\mapsto \log h-\log(1-h),\;\;	\varphi:=\psi^{-1}:\R\rightarrow]0,1[,~y\mapsto
\frac{\exp y}{1+\exp y}.
\end{equation}
Note that $\phi$ is Lipschitz continuous.
We define
\begin{equation}\label{080910d}
\bar{f}:\R\rightarrow \R,~y\mapsto f(\varphi(y))
\end{equation}
and the corresponding cost function
\begin{equation}\label{080910e}
\overline{C}(y,\zeta)=\widetilde{C}(\varphi(y),\varphi(y+\zeta)-\varphi(y)):=\log(1-\delta)+\log(1-C(\varphi(y),\varphi(y+\zeta))).
\end{equation}

Admissible trading strategies are now just common impulse control strategies and hence independent of the initial values $v_0,h_0$.
Let $v_0>0,\,\widetilde{v}_0>0,\,h_0\in[0,1]$ and $y_0\in\R$. It is not hard to see that the optimal growth rates $\rho_{v_0,h_0}$
and $\widetilde{\rho}_{\widetilde{v}_0,y_0}$
are independent of the initial values $v_0,\widetilde{v}_0,h_0,y_0$ and satisfy
$\rho_{v_0,h_0}=\widetilde{\rho}_{\widetilde{v}_0,y_0}.$
For some $v_0>0,\,y_0\in\R^d$, it suffices to maximize over all admissible strategies $\overline{K}=(\tau_n,\zeta_n)_{n\in\N_0}$
\begin{equation*}
\widetilde{J}_{v_0,y_0}\bigl(\overline{K}\bigr)=r+\liminf\limits_{T\to\infty}
\frac{1}{T}E\Biggl(\int\limits_0^T\widetilde{f}(y_{s-})ds+
\sum\limits_{k=0}^\infty\overline{C}(y_{\tau_k-},\zeta_k)
\mathbbm{1}_{\{\tau_k\leq T\}}\Biggr).
\end{equation*}


\subsection{Optimal Strategies}
The maximization problem \eqref{020910f} was solved by Irle and Sass in \cite{IS06}. 
Our main purpose in this subsection is a brief introduction of the approach and the results from \cite{IS06} that are needed later. 
The main tool is the Hamilton-Jacobi-Bellman approach, i.e., they explicitly solved the quasi-variational inequality (QVI)
\begin{equation}\label{121011s}
\max\bigl\{Du+f-l,\M u-u\bigr\}=0,
\end{equation}
where
$\M u(x):=\sup\limits_{y\in[0,1]} u(y)+\widetilde{C}(x,y)$ denotes the maximum operator and 
\begin{align}
D&:=x(1-x)\bigl(\mu-r-\sigma^2x\bigr)\frac{d}{dx}
+\frac{1}{2}\sigma^2x^2(1-x)^2\frac{d^2}{dx^2}.\label{eq:generator}
\end{align}
\begin{defi}\label{121011l}
	An impulse control strategy $\widetilde{K}=(\tau_n,\xi_n)_{n\in\N_0}$ is a \textit{(proportional) constant boundary strategy} if there exist $a<\alpha\leq\beta<b\in\, ]0,1[$ such that $\xi_0\in\,]a,b[$ and
	\begin{equation*}
	\tau_{n}=\inf\bigl\{t> \tau_{n-1}\,:\,h_t^{n-1}\notin\,]a,b[\bigr\},
	~~\xi_{n}=
	\begin{cases}
	\alpha, &h_{\tau_{n}}^{n-1}=a,\\
	\beta, &h_{\tau_{n}}^{n-1}=b,
	\end{cases}
	\end{equation*}
	where $h^{n-1}$ is the corresponding risky fraction process.
	We will refer to $\widetilde{K}$ as a \textit{proportional constant boundary strategy given by} $(a,\alpha,\beta,b)$.
\end{defi}


%

In order to find the solution $u$, a modified version of the cost function is used in \cite{IS06}.

\begin{defi}\label{121011i}
	We define the modified cost function $\Gamma$ for all $x,y\in\,]0,1[$ by
	\begin{equation}\label{121011m}
	\Gamma(x,y):=
	\begin{cases}
	\log\left(\frac{1-\delta+\gamma x}
	{1+\gamma y} \right), &y> x,\\
	\log\left(\frac{1-\delta-\gamma x}
	{1-\gamma y}\right), &y\leq x.
	\end{cases}
	\end{equation}
\end{defi}


Now we are able to state a collection of the results from \cite{IS06}, which we will use as a reference.

\begin{satz}\label{200911h}
	Let $\delta\in\,]0,1[,\,\gamma\in[0,1-\delta[$, $D$ as in \eqref{eq:generator},
	\begin{equation}\label{200911j}
	f:[0,1]\rightarrow\R,~h\mapsto -\frac{1}{2}\sigma^2h^2+(\mu-r)h,
	\end{equation}
	$\Gamma$ as in \eqref{121011m} and let the Merton fraction $\hat{h}:=\frac{\mu-r}{\sigma^2}$ satisfy $\hat{h}\in\,]0,1[$.
	Then there exist constants $l>0$ and $a<\alpha\leq x_0\leq\beta<b\in\,]0,1[$ such that the function $u\in C^1([0,1],\R)$ defined by
	
	\begin{equation}
	u(x):=
	\begin{cases}
	\Gamma(x,\alpha), &x\leq a,\\
	u(a)+\int\limits_a^xg(y,x_0,l)dy, &a<x\leq b,\\
	u(\beta)+\Gamma(x,\beta), &x>b,
	\end{cases}
	\end{equation}
	where
	\begin{equation}\label{210911a}
	g(x,x_0,l):=\begin{cases}
	\frac{\left(\frac{1-x}{x}\right)^{2\hat{\eta}-1}}
	{x(1-x)f(1)}
	\Bigl(\bigl(l-xf(1)\bigr)
	\bigl(\frac{x}{1-x}\bigr)^{2\hat{\eta}-1}
	-\bigl(l-x_0f(1)\bigr)
	\bigl(\frac{x_0}{1-x_0}\bigr)^{2\hat{\eta}-1}
	\Bigr), &\hat{\eta}\neq\frac{1}{2},\\
	\frac{1}{x(1-x)}\int\limits_{x_0}^x
	\bigl(\frac{l}{\sigma^2}-\frac{y}{2}(1-y)\bigr)
	\frac{2}{y(1-y)}dy, &\hat{\eta}=\frac{1}{2},
	\end{cases}
	\end{equation}
	has the following properties:
	\begin{enumerate}[{(}i{)}]
		\item Every constant boundary strategy $\widetilde{K}$ given by the constants $(a,\alpha,\beta,b)$ in the sense of Definition \ref{121011l} is optimal for the modified cost function $\Gamma$ with the optimal value
		$\tilde{J}\bigl(\widetilde{K}\bigr)=r+l=\rho$.
		\item If $a\leq\alpha(1-\delta)$, then $\widetilde{K}$ is also optimal for the original cost function $\widetilde{C}$.
		\item $Du(x)+f(x)-l\leq 0$ and $\,u(y)-u(x)+\Gamma(x,y)\leq 0$ for all $x,y\in\,[0,1]$.
		\item $Du(x)+f(x)-l=0$ for all $x\in\,]a,b[.$
		\item $u'(\beta)=-\frac{\partial}{\partial y}\Gamma(b,y)|_{y=\beta}=-\frac{\gamma}{1-\gamma\beta}$ and $\,u'(\alpha)=-\frac{\partial}{\partial y}\Gamma(a,y)|_{y=\alpha}=\frac{\gamma}{1+\gamma\alpha}$.	
		\item $u'(b)=\frac{\partial}{\partial x}\Gamma(x,\beta)|_{x=b}=-\frac{\gamma}{1-\delta-\gamma b}$ and $\,u'(a)=\frac{\partial}{\partial x}\Gamma(x,\alpha)|_{x=a}=\frac{\gamma}{1-\delta +\gamma a}$.	
		\item $\frac{\partial}{\partial x}g(x,x_0,l)|_{x=\beta}<
		-\frac{\gamma^2}{(1-\gamma\beta)^2}$ and $\,
		\frac{\partial}{\partial x}g(x,x_0,l)|_{x=b}>
		-\frac{\gamma^2}{(1-\delta-\gamma b)^2}.$
		\item $\frac{\partial}{\partial x}g(x,x_0,l)|_{x=\alpha}<
		-\frac{\gamma^2}{(1+\gamma\alpha)^2}$ and $\,
		\frac{\partial}{\partial x}g(x,x_0,l)|_{x=a}>
		-\frac{\gamma^2}{(1-\delta+\gamma a)^2}.$
		\item $g(x,x_0,l)<\frac{\gamma}{1-\delta +\gamma x}$ on $]0,a[$ and 			$\,g(x,x_0,l)>-\frac{\gamma}{1-\delta -\gamma x}$ on $]b,1[$.
		\item $\alpha=\beta$ for $\gamma=0$ and $\alpha<x_0<\beta$ for $\gamma>0$.
	\end{enumerate}
\end{satz}
\begin{proof}
	Cf.$~$\cite{IS06}, pp. 929-936 and Remark 9.1.\vspace{-3mm}
	
\end{proof}

\subsection{Optimal strategies in the transformed space}\label{251011a}
If we define $\bar{f}:=f\circ\varphi\,$ and the cost function $\overline{C}$ as in \eqref{080910e}, then every admissible strategy $\overline{K}=(\tau_n,\xi_n)_{n\in\N_0}\in\bar{\mathcal{A}}$ yields the expected growth rate
\begin{equation}\label{121011q}
\bar{J}\bigl(\overline{K}\bigr)
=r+\liminf\limits_{T\to\infty}\frac{1}{T}
E\Biggl(\int\limits_0^T\bar{f}(Y_{s-})ds+
\sum\limits_{k=0}^\infty\overline{C}(Y_{\tau_k-},\xi_k)
\mathbbm{1}_{\{\tau_k\leq T\}}\Biggr),
\end{equation}
where $Y=\psi(h)$ and $h$ is the corresponding risky fraction process. 
Writing 
\begin{align}
\overline{\M }v(x):=\sup\limits_{y\in\R}v(y)+\overline{C}(x,y-x),\;\;\Lu &:=\frac{1}{2}\sigma^2\frac{d^2}{dx^2}
+\Bigl(\mu-r-\frac{1}{2}\sigma^2\Bigr)\frac{d}{dx},\label{131011b}
\end{align}
we get the following connection:
If a smooth enough function $u:[0,1]\rightarrow \R$ is a solution to the QVI \eqref{121011s} on $[0,1]$, then $u\circ\varphi$ is a solution to the QVI
\begin{equation}\label{121011t}
\max\bigl\{\Lu  v+\bar{f}-l,\overline{\M }v-v\bigr\}=0
\end{equation}
on $\R$. Using that we can translate the results presented before to $\R$ with a corresponding notion of a constant boundary strategy.
The next Theorem is mainly the counterpart of Theorem \ref{200911h}. We omit the straightforward proof.

\begin{satz}\label{210611a}
	Let $\delta\in\,]0,1[,\,\gamma\in[0,1-\delta[$, $\Lu $ as in \eqref{131011b}, $\bar{f}=f\circ\varphi$ as in \eqref{080910d}, $\Gamma$ as in \eqref{121011m} and let the Merton fraction $\hat{h}:=\frac{\mu-r}{\sigma^2}$ satisfy $\hat{h}\in\,]0,1[$.
	Then there exist constants $l>0$ and $a<\alpha\leq\beta<b\in\R$ and a bounded continuous function $u:\R\rightarrow[0,\infty[$ satisfying
	\begin{enumerate}[{(}i{)}]
		\item $\alpha=\beta~\Leftrightarrow~\gamma=0$.
		\item Every constant boundary strategy $\overline{K}=\overline{K}(a,\alpha,\beta,b)=(\tau_n,\xi_n)_{n\in\N_0}$ is optimal for the modified cost function $\Gamma$ with optimal value $\bar{J}\bigl(\overline{K}\bigr)=r+l=\rho$.
		\item If $\varphi(a)\leq\varphi(\alpha)(1-\delta)$, then $\overline{K}$ is also optimal for the original cost function $\overline{\Gamma}$.
		\item $\max\{f(0),f(1)\}<l<f(\hat{h})$.
		\item For every $y\in\,]a,b[$ and $\sigma\leq\tau^\ast(y):=
		\inf\bigl\{t>0: X^y_t\notin\,]a,b[\bigr\}$, where $X^y$ is the diffusion corresponding to \eqref{131011b} starting in $y$, it holds $u(y)=E\left(\int\limits_0^\sigma \bigl(\bar{f}(X^y_s)-l\bigr)ds + u(X^y_\sigma)\right)$.
		\item $\|u\|_\infty\leq\|\Gamma\|_\infty<\infty,~u|_{[a,b]}\in C^2([a,b]),~u|_{\R\backslash[a,b]}\in C^2(\R\backslash [a,b])$.
		\item There is an $x_0\in\bigl\{y\in\R:\bar{f}(y)-l\geq 0\bigr\}\cap[\alpha,\beta]$ such that $u(x_0)=\|u\|_\infty$ and $\alpha<x_0<\beta~\Leftrightarrow~\gamma>0$.
		\item $\Lu u(x)+\bar{f}(x)-l\leq 0$ and $u(y)-u(x)+\Gamma(x,y)\leq 0$ for all $x,y\in\R$.
		\item $\Lu u(x)+\bar{f}(x)-l=0$ for all $x\in\,]a,b[.$
		\item $u(\beta)-u(x)+\Gamma(x,\beta)=0$ for all $x\geq b$ and $u(\alpha)-u(x)+\Gamma(x,\alpha)=0$ for all $x\leq a$.
		\item $u'(\beta)=-\frac{\partial}{\partial y}\Gamma(b,y)|_{y=\beta}=\frac{-\gamma\varphi'(\beta)}{1-\gamma\varphi(\beta)}$ and $u'(\alpha)=-\frac{\partial}{\partial y}\Gamma(a,y)|_{y=\alpha}=\frac{\gamma\varphi'(\alpha)}{1+\gamma\varphi(\alpha)}$.
		\item $u'(b)=\frac{\partial}{\partial x}\Gamma(x,\beta)|_{x=b}=\frac{-\gamma\varphi'(b)}{1-\delta-\gamma\varphi(b)}$ and $u'(a)=\frac{\partial}{\partial x}\Gamma(x,\alpha)|_{x=a}=\frac{\gamma\varphi'(a)}{1-\delta+\gamma\varphi(a)}$.
		\item For $\gamma>0$ $u$ is strictly increasing on $(-\infty,0[$ and strictly decreasing on $]0,\infty[$.
		\item For $\gamma=0$ $u$ is strictly increasing on $[a,\alpha]$, strictly decreasing on $[\alpha,b]$ and $u\equiv 0$ on $\R\setminus [a,b]$.
		\item $u\circ\psi$ differs only by a constant from its counterpart in Theorem \ref{200911h} while $l$,$\,\varphi(x_0)$, 																		$\varphi(a)$, $\,\varphi(\alpha)$, $\,\varphi(\beta)$ and $\varphi(b)$ are equal.
	\end{enumerate}
\end{satz}


\section{A model with only proportional transaction costs}\label{280911i}

Considering the excluded case $\delta=0$ in the models of the previous section formally leads to a model with pure proportional transaction costs. We, however, then leave the scope of these models, since frequent trading is not punished anymore leaving the impulse control strategies too restrictive. For $\delta=0$ continuous trading has to be considered and the optimal strategy is to keep the risky fraction process in some bounded interval $[A,B]\subset\,]0,1[$ with minimal effort, i.e.$~$the risky fraction process is reflected at the boundaries. 

To establish the convergences in the following section, we begin by introducing the model with only proportional transaction costs 
and show how the corresponding Hamilton-Jacobi-Bellman approach can be used to solve the optimization problem in this section. The most similar but not identical models are that of \cite{TKA88}, where transaction costs are only to be paid for selling stocks, and that of \cite{S97} or \cite{AST01}, where the investor's wealth is considered after liquidation. As, however, the results can basically be obtained using well-known arguments, we do not give the details here and refer the interested reader to \cite{L12}.
%

%

Let $V,X,Y$ and $h$ denote the processes describing the portfolio 
and suppose fixed starting values $(x_0,y_0)\in[0,\infty[^2$ with $v_0=x_0+y_0>0,~h_0=\frac{y_0}{v_0}\in\Delta$ are given.
The transaction costs are now given by a constant $\gamma\in\,]0,1[$ denoting the fraction of the transaction volume that has to be paid for every transaction.
Since without fixed transaction costs continuous trading is allowed, we no longer consider only impulse control strategies as trading strategies but càdlàg policies.

\begin{defi}\label{190911b}
	An \textit{investment policy} is a pair $(L,M)$ of nonnegative, nondecreasing and càdlàg processes $L=(L_t)_{t\geq0},M=(M_t)_{t\geq0}$. Here, the cumulative amount of money transferred from bond to stock up to the time $t$ is denoted by $L_t$, whereas $M_t$ is the cumulative amount of money that stocks are sold for up to time $t$.
\end{defi}
For an investment policy $(L,M)$ the portfolio-dynamics is given by
\begin{align}\label{220811a}
\begin{split}
dX_t
&=rX_{s-}ds+(1-\gamma)dM_t-(1+\gamma)dL_t,\;\;
dY_t
=\mu Y_{s-}ds+\sigma Y_{s-}dW_s-dM_t+dL_t.
\end{split}
\end{align}
We therefore call an investment policy $(L,M)$ \textit{admissible for }$(v_0,h_0)$ if there exist unique processes $(X_t)_{t\geq 0},(Y_t)_{t\geq 0}$ that satisfy \eqref{220811a}, $X_t,Y_t\geq 0$ and $V_t=X_t+Y_t>0$ for all $t\geq 0$. To clarify the dependence we will write $X^{(L,M)},Y^{(L,M)}$ and $V^{(L,M)}$ if needed.\\
Furthermore, we denote the set of all admissible investment policies as
$\mathcal{A}_{v_0,h_0}$.
The objective is again to find an optimal admissible policy $(L^*,M^*)$ that maximizes the expected growth rate
\begin{equation}\label{220811d}
\rho_{v_0,h_0}:=\sup\limits_{(L,M)\in\mathcal{A}_{v_0,h_0}}
J_{v_0,h_0}(L,M):=\sup\limits_{(L,M)\in\mathcal{A}_{v_0,h_0}}
\liminf\limits_{T\to\infty}\frac{1}{T}E\Bigl(\log\bigl(V_T^{(L,M)}\bigr)\Bigr).
\end{equation}

\begin{defi}\label{200911g}
	A continuous investment policy $(L,M)\in\mathcal{A}_{v_0,h_0}$ with $L_0=M_0=0$ is a \textit{control limit policy} for the limits $A<B\in\,]0,1[$ if the corresponding risky fraction process $h$ satisfies
	\begin{equation}\label{200911a}
	L_t=\int\limits_0^t\mathbbm{1}_{\{h_s=A\}}dL_s,~
	M_t=\int\limits_0^t\mathbbm{1}_{\{h_s=B\}}dM_s,~h_t\in[A,B]\mbox{
		for all }t\geq 0.
	\end{equation}
\end{defi}
The existence and uniqueness of control limit policies was proved in Theorem 9.2 in \cite{SS94}.\\
Finally, we can establish a verification theorem.
\begin{satz}\label{200911f}
	Let $f$ be as in \eqref{200911j} and let $D$ denote the generator from \eqref{eq:generator}. Assume there exists some $v\in C^2([0,1],\R)$ and $l\in\R$ such that for all $x\in [0,1]$
	\begin{enumerate}[{(}a{)}]
		\item $Dv(x)\leq -f(x)+l$,
		\item $v'(x)\leq \frac{\gamma}{1+x\gamma}$,
		\item $v'(x)\geq -\frac{\gamma}{1-x\gamma}.$
	\end{enumerate}
	Let further $A<B\in \,]0,1[$ and $h_0\in [A,B]$. Suppose $v$ additionally satisfies
	\begin{enumerate}[{(}i{)}]
		\item $Dv(x)= -f(x)+l$ for all $x\in[A,B],$
		\item $v'(x)= \frac{\gamma}{1+x\gamma}$ for all $x\in[0,A],$
		\item $v'(x)= -\frac{\gamma}{1-x\gamma}$ for all $x\in[B,1]$.
	\end{enumerate}
	Then the control limit policy $(L,M)$ for the limits $A,B$ is optimal.	
\end{satz}

%

Now it remains to prove the existence and uniqueness of constants $A<B\in\,]0,1[$ such that the control limit policy for the limits $A,B$ is optimal. Their existence will be established in Proposition \ref{200911e} of the next section as the limit of the boundaries $(a_n,\alpha_n,\beta_n,b_n)$ of the constant boundary strategies for fixed costs $\delta_n\rightarrow 0$.
The following result is essentially used in the next section to establish the uniqueness of the optimal $A,B$.
\begin{lem}\label{280911d}
	Let $A<B\in\,]0,1[$ and $h_0\in [A,B]$. Suppose the control limit policy $(L,M)$ for the limits $A,B$ is optimal with the optimal value $\rho_{v_0,h_0}>r$. Then the corresponding value process $V$ satisfies
	\begin{equation*}
	\limsup\limits_{T\to\infty}\frac{1}{T} E\Biggl(
	\int\limits_0^T\frac{1}{V_s}dL_s\Biggr)>0,~\,
	\limsup\limits_{T\to\infty}\frac{1}{T} E\Biggl(
	\int\limits_0^T\frac{1}{V_s}d M_s\Biggr)>0.
	\end{equation*}
\end{lem}

\begin{proof}
	It holds that	
	\begin{equation*}
	\rho_{v_0,h_0}=
	J_{v_0,h_0}(L,M)=r+\liminf\limits_{T\to\infty}\frac{1}{T}E
	\Biggl(-\int\limits_0^T\frac{1+\gamma}{1-h_s}\frac{1}{V_s}dL_s
	+\int\limits_0^T\frac{1-\gamma}{1-h_s}\frac{1}{V_s}dM_s\Biggr)
	\end{equation*}
	and thus for $l:=\rho_{v_0,h_0}-r$
	\begin{align*}
	0<l&=\liminf\limits_{T\to\infty}\frac{1}{T}E
	\Biggl(-\int\limits_0^T\frac{1+\gamma}{1-h_s}\frac{1}{V_s}dL_s
	+\int\limits_0^T\frac{1-\gamma}{1-h_s}\frac{1}{V_s}dM_s\Biggr)\\
	&\leq
	\limsup\limits_{T\to\infty}\frac{1}{T}E
	\Biggl(\int\limits_0^T\frac{1-\gamma}{1-h_s}\frac{1}{V_s}dM_s\Biggr)
	=\frac{1-\gamma}{1-B}\cdot
	\limsup\limits_{T\to\infty}\frac{1}{T}E
	\Biggl(\int\limits_0^T\frac{1}{V_s}dM_s\Biggr).
	\end{align*}
	Now we assume
	\begin{equation}\label{270911j}
	\limsup\limits_{T\to\infty}\frac{1}{T} E\Biggl(
	\int\limits_0^T\frac{1}{V_s}dL_s\Biggr)=0
	\end{equation}
	and consider new limits $\tilde{A}<\tilde{B}\in\,]0,1[$ and a new starting value $\tilde{h_0}\in [\tilde{A},\tilde{B}]$ such that after the transformation from $]0,1[$ to $\R$ via $\psi(x):=\log\frac{x}{1-x}$ the distances remain the same, i.\,e.
	\begin{equation}\label{270911i}
	\psi(B)-\psi(A)=\psi(\tilde{B})-\psi(\tilde{A}),~
	\psi(B)-\psi(h_0)=\psi(\tilde{B})-\psi(\tilde{h_0})
	\end{equation}
	and $\tilde{B}>B$. Furthermore, we denote by $(\tilde{L},\tilde{M})$ the control limit policy for the new limits $\tilde{A},\tilde{B}$. 
	The processes $y:=\psi(h)$, where $h$ is the risky fraction processes corresponding to $(L,M)$, satisfies
	\begin{equation}\label{270911g}
	y_t=y_0+\Bigl(\mu-r-\frac{\sigma^2}{2}\Bigr)\cdot t
	+\sigma W_t+\frac{1+\gamma A}{A(1-A)}\int\limits_0^t
	\frac{1}{V_s}d L_s
	-\frac{1-\gamma B}{B(1-B)}\int\limits_0^t
	\frac{1}{V_s}d M_s,
	\end{equation}
	The corresponding formula for $\tilde{y}:=\psi(\tilde{h})$ associated to $(\tilde{L},\tilde{M})$ yields that both are instantaneous reflections of the same diffusion with different boundaries and different starting values. Since the processes involved in an instantaneous reflection are unique (cf.$~$\cite{GS72}, p.$~$185
	), \eqref{270911i} yields
	\begin{equation*}
	\frac{1+\gamma A}{A(1-A)}\int\limits_0^t
	\frac{1}{V_s}d L_s=\frac{1+\gamma \tilde{A}}{\tilde{A}(1-\tilde{A})}\int\limits_0^t
	\frac{1}{\tilde{V}_s}d \tilde{L}_s,~
	\frac{1-\gamma B}{B(1-B)}\int\limits_0^t
	\frac{1}{V_s}d M_s=\frac{1-\gamma \tilde{B}}{\tilde{B}(1-\tilde{B})}\int\limits_0^t
	\frac{1}{\tilde{V}_s}d \tilde{M}_s.
	\end{equation*} 	
	Thus \eqref{270911j} implies
	$$\limsup\limits_{T\to\infty}\frac{1}{T} E\Biggl(
	\int\limits_0^T\frac{1}{\tilde{V}_s}d\tilde{L}_s\Biggr)=0$$
	and we therefore get 
	the contradiction
	\begin{align*}
	J_{v_0,\tilde{h}_0}\bigl(\tilde{L},\tilde{M}\bigr)
	&=r+\liminf\limits_{T\to\infty}\frac{1}{T}E
	\Biggl(-\int\limits_0^T\frac{1+\gamma}
	{1-\tilde{h}_s}\frac{1}{\tilde{V}_s}d\tilde{L}_s
	+\int\limits_0^T\frac{1-\gamma}
	{1-\tilde{h}_s}\frac{1}{\tilde{V}_s}d\tilde{M}_s\Biggr)\\
	&=r+\liminf\limits_{T\to\infty}\frac{1}{T}E
	\Biggl(\frac{1-\gamma}{1-B}\int\limits_0^T\frac{1}{V_s}d M_s\Biggr)
	\frac{1-\gamma B}{B}\frac{\tilde{B}}{1-\gamma \tilde{B}}
	>r+l=\rho_{v_0,h_0}.
	\end{align*}
\end{proof} 


\section{Convergence in case of vanishing fixed costs}\label{nofix}
We now prove the convergence of the boundaries and that of the optimal expected growth rates for vanishing fixed transaction costs. Moreover, the convergence of the corresponding QVI-solutions is proved, yielding the uniqueness and optimality of the limiting constants $A,B$ in the model of the previous section.

\subsection{Convergence of boundaries and optimality}
In this subsection, we first establish the convergence 
$$\lim\limits_{\delta\to 0}a(\delta)=A=\lim\limits_{\delta\to 0}\alpha(\delta)\mbox{ and }
\lim\limits_{\delta\to 0}\beta(\delta)=B=\lim\limits_{\delta\to 0}b(\delta)$$
to some constants $A,B\in\,]0,1[$. This is achieved on $\R$ via the coordinate transformation $\psi$. Then we prove the convergence of the corresponding optimal growth rates and that of the QVI-solutions $u(\delta)$ from Theorem \ref{200911h} to some solution $u$ to the HJB-equation from Theorem \ref{200911f} with the above $A,B$. Using that, we finally obtain uniqueness and optimality of $A,B$.
\begin{prop}\label{240611a}
	Let $(\delta_n)_{n\in\N}\in\,]0,1[^\N$ be a sequence of trading proportions defining the fixed costs in the sense of \eqref{cost} for a fixed common $\gamma>0$ defining the proportional costs and satisfying $\gamma<1-\sup\limits_{n\in\N}\delta_n$. Let the Merton fraction $\hat{h}:=\frac{\mu-r}{\sigma^2}$ satisfy $\hat{h}\in\,]0,1[$ and suppose $\lim\limits_{n\to\infty}\delta_n=0$. Then the corresponding constants $a_n,\alpha_n,\beta_n$ and $b_n$ given by Theorem \ref{210611a} for every $n\in\N$ satisfy
	\begin{equation*}
	\lim\limits_{n\to\infty}
	|b_n-\beta_n|=0\mbox{ and }\lim\limits_{n\to\infty}|a_n-\alpha_n|=0.
	\end{equation*}
\end{prop}

\begin{proof}
	Since $\lim\limits_{n\to\infty}|a_n-\alpha_n|=0$ can be shown the same way, we restrict ourselves to proving only $\lim\limits_{n\to\infty}|b_n-\beta_n|=0$ here. For that purpose we define for every $n\in\N$ the function
	\begin{equation*}
	g_n:I_n:=[\beta_n,b_n]\rightarrow \R,~x\mapsto
	u_n(x)-u_n(\beta_n)-\Gamma_n(x,\beta_n),
	\end{equation*}
	where $u_n$ is the solution from Theorem \ref{210611a} and $\Gamma_n$ is the cost function with the corresponding constants $\gamma,\delta_n$ and the transformation function $\varphi$.
	
	The proof is based on the following idea: The assumption $\inf\limits_{n\in\N}|\beta_n-b_n|>0$ leads to $\lim\limits_{n\to\infty}g_n=\lim\limits_{n\to\infty}g_n'=\lim\limits_{n\to\infty}g_n''=0$ uniformly on a common interval $I\subseteq \,]\beta_n,b_n[$, $n\in\N$. But since we have $Lu_n=-\bigl(\bar{f}-l_n\bigr)$ on $I$ for every $n\in\N$, the limit cost function $v(x)=-\Gamma_0(x,y)$ with $\delta=0$ has to satisfy $Lv=-\bigl(\bar{f}-l_0\bigr)$ on $I$ for some fixed $y\in\R$, which is a contradiction.\leer
	
	a) Let $n\in\N$. What we show here, is that for all $x\in [\beta_n,b_n]$
	\begin{equation}\label{280611a}
	0\leq g_n(x)\leq \left|\log\left(\frac{1-\delta_n-\gamma}{1-\gamma}\right)\right|.
	\end{equation}
	Since for every $x\in\,]\beta_n,b_n]$ the inequality $$u_n(\beta_n)+\Gamma_n(x,\beta_n)<u_n(x)+\Gamma_n(x,x)$$ would imply
	$$u_n(b_n)\stackrel{\ref{210611a} (x)}{=}
	u_n(\beta_n)+\Gamma_n(b_n,\beta_n)\stackrel{\eqref{200611f}}{<}
	u_n(x)+\Gamma_n(b_n,x)\stackrel{ }{\leq} u_n(b_n),$$ 
	hence a contradiction, where we used
	\begin{equation}\label{200611f}
	\frac{\partial}{\partial x}\Gamma(x,y)
	=\frac{-\gamma\varphi^{'}(x)}{1-\delta-\gamma\varphi(x)}<0\mbox{ and }
	\frac{\partial}{\partial y}\Gamma(x,y)
	=\frac{\gamma\varphi^{'}(y)}{1-\gamma\varphi(y)}>0.
	\end{equation}
	Therefore,
	$$u_n(x)\stackrel{\ref{210611a}(viii)}{\geq} u_n(\beta_n)+\Gamma_n(x,\beta_n)\geq u_n(x)+\Gamma_n(x,x)
	=u_n(x)+\log\left(\frac{1-\delta_n-\gamma\varphi(x)}
	{1-\gamma\varphi(x)}\right).$$
	But since $\Gamma_n(x,x)$ is strictly decreasing in $x$ on $\R$, \eqref{280611a} follows.
	
	b) Now let $\varepsilon>0$ and suppose $|\beta_n-b_n|\geq 4\varepsilon$
	for infinitely many $n\in\N$ and without loss of generality for all $n\in\N$.
	
	b1) We first show that there is a constant $c=c(\mu,\sigma^2,r,\gamma)$ such that
	\begin{equation}\label{300611a}
	\sup\limits_{n\in\N}\sup\limits_{x\in [\beta_n,\beta_n+2\varepsilon]}|g_n''(x)|\leq c.
	\end{equation}
	By \eqref{200611f} we have for $n\in\N$ and $x\in\,]\beta_n,b_n[$
	\begin{equation}\label{300611b}
	g_n'(x)=u_n'(x)+\frac{\gamma\varphi'(x)}{1-\delta_n-\gamma\varphi(x)}
	\end{equation}
	and hence
	\begin{equation}\label{300611c}
	g_n''(x)=u_n''(x)+\frac{\gamma\varphi''(x)}
	{1-\delta_n-\gamma\varphi(x)}
	+\left(\frac{\gamma\varphi'(x)}{1-\delta_n-\gamma\varphi(x)}\right)^2,
	\end{equation}
	where by
	taking derivatives
	\begin{equation}\label{300611d}
	\varphi'(x)=\varphi(x)-\varphi^2(x)\mbox{ and }
	\varphi''(x)=\varphi'(x)(1-2\varphi(x))=\varphi(x)-
	3\varphi^2(x)+2\varphi^3(x).
	\end{equation}
	Since $\varphi(x)\in\,]0,1[$ for all $x\in\R$ and $\gamma\in\,\bigl]0,1-\sup\limits_{m\in\N}\delta_m\bigr[$, it remains by \eqref{300611c} and \eqref{300611d} to prove the boundedness of $u_n''$.
	Since $\mu-r-\frac{1}{2}\sigma^2=0$
	directly implies this, we only consider the case $\mu-r-\frac{1}{2}\sigma^2\neq0$ here.
	It was enough to see the boundedness of $u_n'$ in $n\in\N$ on $[\beta_n,\beta_n+2\varepsilon]$. 
	By Theorem \ref{210611a} $~u_n$ is strictly decreasing on $[\beta_n,b_n]$ and we can find an $x_n<\beta_n$, where $u_n$ attains its maximum on $[a_n,b_n]$.
	In view of \eqref{280611a}, \eqref{300611b} and the mean value theorem applied to $g_n$, we further get for some $\xi\in\,]\beta_n+2\varepsilon,\beta_n+3\varepsilon[$
	\begin{equation}\label{300611e}
	|u_n'(\xi)|\leq \frac{1}{\varepsilon}\left|\log\left(\frac{1-\delta_n-\gamma}{1-\gamma}\right)\right|+
	\frac{\gamma\varphi'(\xi)}{1-\delta_n-\gamma\varphi(\xi)}.
	\end{equation}
	Now in the case, where $u_n'$ attains its minimum on $[x_n,\xi]$ in some $x\in\,]x_n,\xi[$, we have $u_n''(x)=0$ and hence 
	it is easily seen that
	\begin{equation*}\label{300611f}
	\sup\limits_{y\in [\beta_n,\beta_n+2\varepsilon]}|u_n'(y)|
	\leq|u_n'(x)|\leq\left|\frac{-\bar{f}(x)+l_n}
	{\mu-r-\frac{1}{2}\sigma^2}\right|.
	\end{equation*}
	If on the contrary the minimum is attained in $\xi$, we have by \eqref{300611e}
	\begin{equation*}
	\sup\limits_{y\in [\beta_n,\beta_n+2\varepsilon]}|u_n'(y)|
	\leq \frac{1}{\varepsilon}\left|\log\left(\frac{1-\delta_n-\gamma}{1-\gamma}\right)\right|+
	\frac{\gamma\varphi'(\xi)}{1-\delta_n-\gamma\varphi(\xi)}.
	\end{equation*}
	
	b2) Here, we show that there exist $\beta<b\in\R$ such that $I:=[\beta,b]\subseteq \,]\beta_n,b_n[$ for infinitely many $n\in\N$.
	
	Theorem \ref{210611a} yields for the constant $\bar{l}_1:=\inf\limits_{n\in\N}\,l_n$ the boundedness of the set\\ $\bigl\{x\in\R:\bar{f}(x)\geq \bar{l}_1\bigr\}$ in $\R$ and for every $n\in\N$ also
	$x_n\in\bigl\{x\in\R:\bar{f}(x)\geq \bar{l}_1\bigr\}=:[y_0,y_1].$
	Since for all $x,y\in\R$ we have $|u_n(x)-u_n(y)|
	\stackrel{\ref{210611a}(vi)}{\leq}\sup\limits_{n\in\N}
	\|\Gamma_n\|_\infty<\infty$, we can use Theorem \ref{210611a} (v)
	to get
	$y_2:=\sup\limits_{n\in\N}|y_1-b_n|<\infty$
	and thus $[\beta_n,b_n]\subseteq [y_0,y_1+y_2]$ for all $n\in\N$. Hence
	we can find convergent subsequences $\bigl(\beta_{n_k}\bigr)_{k\in\N},\bigl(b_{n_k}\bigr)_{k\in\N}$, which satisfy $|\beta_{n_k}-b_{n_k}|\geq 4\varepsilon$ from our previous assumption.
	
	b3) Now we assume by b2) without loss of generality $I\subseteq \,]\beta_n,b_n[$ for all $n\in\N$.\\
	We then can use a) and b1) together with a Landau-type inequality (see \cite{L14}) to get
	\begin{equation}\label{010711a}
	\lim\limits_{n\to\infty}\left\|g_n'\right\|_{I,\infty}=0.
	\end{equation}
	We define $l_0:=\sup\limits_{n\in\N}l_n$ and get from \eqref{010711a} and \eqref{300611b}
	\begin{equation}\label{010711d}
	\lim\limits_{n\to\infty} u_n''(x)=\frac{2}{\sigma^2}\bigl(-\bar{f}(x)+l_0\bigr)
	+\frac{2}{\sigma^2}\Bigl(\mu-r-\frac{1}{2}\sigma^2\Bigr)
	\frac{\gamma\varphi'(x)}{1-\gamma\varphi(x)},
	\end{equation}
	uniformly in $x\in I$, and by \eqref{300611c} we therefore have uniformly in $x\in I$
	\begin{equation*}
	\lim\limits_{n\to\infty} g_n''(x)=
	\frac{2}{\sigma^2}\bigl(-\bar{f}(x)+l_0\bigr)
	+\frac{2}{\sigma^2}\Bigl(\mu-r-\frac{1}{2}\sigma^2\Bigr)
	\frac{\gamma\varphi'(x)}{1-\gamma\varphi(x)}
	+\frac{\gamma\varphi''(x)}{1-\gamma\varphi(x)}
	+\left(\frac{\gamma\varphi'(x)}{1-\gamma\varphi(x)}\right)^2.
	\end{equation*}
	This implies together with \eqref{010711a} for all $x\in I$
	\begin{equation}\label{010711b}
	0=\frac{2}{\sigma^2}\bigl(-\bar{f}(x)+l_0\bigr)
	+\frac{2}{\sigma^2}\Bigl(\mu-r-\frac{1}{2}\sigma^2\Bigr)
	\frac{\gamma\varphi'(x)}{1-\gamma\varphi(x)}
	+\frac{\gamma\varphi''(x)}{1-\gamma\varphi(x)}
	+\left(\frac{\gamma\varphi'(x)}{1-\gamma\varphi(x)}\right)^2.
	\end{equation}
	
	b4) Here, we show that \eqref{010711b} is not possible.\\
	We replace in \eqref{010711b} the functions $\varphi'$ and $\varphi''$ by \eqref{300611d} and the function $\bar{f}$ by its definition and convert the fraction, yielding
	\begin{equation}\label{010711c}
	\frac{\sum\limits_{i=0}^4\lambda_i\varphi^i(x)}{\bigl(1-\gamma\varphi(x)\bigr)^2}=0
	\end{equation}
	for some $\lambda_i\in\R$, $i\in\{0,1,2,3,4\}$, where $\lambda_0=l_0$. Now since the denominator is positive, we consider only the numerator and convert \eqref{010711c} by using the definition of $\varphi$ for all $x\in I$ to the form
	$\sum_{i=0}^4\mu_i (e^x)^i=0$
	for some $\mu_i\in\R$, $i\in\{0,1,2,3,4\}$, where $\mu_0=\lambda_0=l_0$.
	Therefore we get $\sum_{i=0}^4\mu_i x^i=0$ for all $x\in [e^\beta,e^b]$.
	Finally, we deduce by taking derivatives $0=\mu_0=l_0\geq \inf\limits_{n\in\N}l_n
	>f(0)=0$ and hence a contradiction.
\end{proof}

Using the result above we can prove the convergence of the boundaries at least for some subsequence.

\begin{prop}\label{200911b}
	Let $(\delta_n)_{n\in\N}\in\,]0,1[^\N$ be a sequence of trading proportions defining the fixed costs in the sense of \eqref{cost} for a fixed common $\gamma>0$ defining the proportional costs and satisfying $\gamma<1-\sup\limits_{n\in\N}\delta_n$. Let the Merton fraction $\hat{h}:=\frac{\mu-r}{\sigma^2}$ satisfy $\hat{h}\in\,]0,1[$ and suppose $\lim\limits_{n\to\infty}\delta_n=0$. Then there exist constants $a_0<x_0<b_0\in\R$ and a subsequence $\bigl(\delta_{n_k}\bigr)_{k\in\N}$ such that the corresponding constants $x_{n_k},a_{n_k},\alpha_{n_k},\beta_{n_k}$ and $b_{n_k}$ given by Theorem \ref{210611a} for every $k\in\N$ satisfy
	\begin{equation*}
	\lim\limits_{k\to\infty}x_{n_k}=x_0,~
	\lim\limits_{k\to\infty}a_{n_k}=a_0
	=\lim\limits_{k\to\infty}\alpha_{n_k},~
	\lim\limits_{k\to\infty}\beta_{n_k}=b_0
	=\lim\limits_{k\to\infty}b_{n_k}.
	\end{equation*}
\end{prop}

\begin{proof}
	Theorem \ref{210611a} yields a function $u_n$ and a constant $l_n$ such that
	$$x_n\in\bigl\{y\in\R:\bar{f}(y)-l_n\geq 0\bigr\}\cap [\alpha_n,\beta_n]$$
	and $\sup\limits_{n\in\N}\|u_n\|_\infty\leq\sup\limits_{n\in\N}\|\Gamma_n\|_\infty<\infty.$
	Therefore we can use Theorem \ref{210611a} (v) 
	to get
	\begin{equation}\label{200911d}
	\sup\limits_{n\in\N}b_n<\infty,~\inf\limits_{n\in\N}a_n>-\infty.
	\end{equation}
	Now as in the proof of Proposition \ref{240611a} we have
	\begin{equation*}
	\sup\limits_{n\in\N,x\in[\alpha_n,\beta_n]}|u_n'(x)|<\infty
	\end{equation*}
	and due to Theorem \ref{210611a} (ix)
	for all $x\in \,]a_n,b_n[$
	\begin{equation*}
	u_n''(x)=\frac{2}{\sigma^2}\bigl(-\bar{f}(x)+l_n\bigr)
	-\frac{2}{\sigma^2}\Bigl(\mu-r-\frac{1}{2}\sigma^2\Bigr)u_n'(x)
	\end{equation*}
	and therefore
	\begin{equation}\label{200911c}
	\sup\limits_{n\in\N,x\in[\alpha_n,\beta_n]}|u_n''(x)|<\infty.
	\end{equation}
	Furthermore, since $u_n'(x_n)=0$ and $u_n'(\beta_n)=\frac{-\gamma}{1-\gamma\beta_n}$ by Theorem \ref{210611a}, the mean value theorem together with \eqref{200911c} and \eqref{200911d} yields
	\begin{equation}\label{221011c}
	\inf\limits_{n\in\N}|x_n-\beta_n|>0\mbox{ and analogously }\inf\limits_{n\in\N}|x_n-\alpha_n|>0.
	\end{equation}
	Finally, Proposition \ref{240611a} implies
	\begin{equation*}
	\lim\limits_{n\to\infty}
	|b_n-\beta_n|=0\mbox{ and }\lim\limits_{n\to\infty}|a_n-\alpha_n|=0
	\end{equation*}
	and therefore the existence of the desired constants $a_0<x_0<b_0\in\R$ and of the subsequence $\bigl(\delta_{n_k}\bigr)_{k\in\N}$ follows.
	
\end{proof}

\begin{bem}
	The convergence $\lim\limits_{n\to\infty}|b_n-\beta_n|=0$ and $\lim\limits_{n\to\infty}|a_n-\alpha_n|=0$ shown above is a convergence on $\R$ for the transformed boundaries and is therefore a priori stronger than a convergence on $]0,1[$ before the transformation. It is equivalent if these sequences stay away from the boundary, since the transformation function $\psi$ is Lipschitz continuous on every compact subset of $]0,1[$, which is an easy consequence of $\psi\in C^\infty(]0,1[,\R)$, and its inverse $\varphi$ is also Lipschitz continuous.
	That these sequences in fact stay away from the boundary, can be seen in \eqref{200911d}.
\end{bem}

The next proposition is crucial. It guarantees the existence of a solution $u$ to the HJB-equation described in Theorem \ref{200911f}, which is the limit of the solutions $u_n$ for vanishing fixed costs $\delta_n$. It will also be used for the uniqueness and optimality of the constants $a_0,b_0$ of Proposition \ref{200911b}.


\begin{prop}\label{200911e}
	Let $(\delta_n)_{n\in\N}\in\,]0,1[^\N$ be a sequence of trading proportions defining the fixed costs in the sense of \eqref{cost} for a fixed common $\gamma>0$ defining the proportional costs and satisfying $\gamma<1-\sup\limits_{n\in\N}\delta_n$. Let the Merton fraction $\hat{h}:=\frac{\mu-r}{\sigma^2}$ satisfy $\hat{h}\in\,]0,1[$ and suppose $\lim\limits_{n\to\infty}\delta_n=0$. Let further $a_0<b_0\in\R$ and $\bigl(\delta_{n_k}\bigr)_{k\in\N}$ be from Proposition \ref{200911b} and $u_{n_k},\rho_{n_k}$ be the corresponding functions and optimal growth rates from Theorem \ref{200911h}.
	Then there exists a function $u\in C^2([0,1],\R)$ given by $u(x)=\lim\limits_{k\to\infty}u_{n_k}(x)$ for all $x\in[0,1]$ that satisfies the conditions of Theorem \ref{200911f} for the constants $A:=\varphi(a_0)<B:=\varphi(b_0)$ and $l=\lim\limits_{k\to\infty}\rho_{n_k}-r$.
\end{prop}

\begin{proof}
	Let $l_n,a_n,\alpha_n,x_{0,n},\beta_n$ and $b_n$ denote the constants and $u_n,\Gamma_n$ denote the functions given by Theorem \ref{200911h} for every $n\in\N$ and $\delta_n$. Without loss of generality we have
	\begin{equation}\label{210911b}
	\lim\limits_{n\to\infty}x_{0,n}=x_0,~
	\lim\limits_{n\to\infty}a_{n}=A
	=\lim\limits_{n\to\infty}\alpha_{n},~
	\lim\limits_{n\to\infty}\beta_{n}=B
	=\lim\limits_{n\to\infty}b_{n}
	\end{equation}
	for some $0<A<x_0<B<1$ by Proposition \ref{200911b}. Since for decreasing $\delta_n$ the corresponding optimal growth rates $\rho_n=r+l_n$ are increasing, we define $$l_0:=\lim\limits_{n\to\infty}l_n=\sup\limits_{n\to\infty}l_n.$$
	a) Now we define the function $u$ on $[0,1]$ by
	\begin{equation}\label{270911k}
	u(x):=
	\begin{cases}
	\Gamma_0(x,A), &x\leq A,\\
	u(A)+\int\limits_A^xg(y,x_0,l_0)dy, &A<x\leq B,\\
	u(B)+\Gamma_0(x,B), &x>B,
	\end{cases}
	\end{equation}
	where $\Gamma_0$ is the cost function
	for $\delta=0$ and $g$ is as in \eqref{210911a}. By the definition of $g$ and \eqref{210911b} we have
	$$\lim\limits_{n\to\infty}g(x,x_{0,n},l_n)=g(x,x_0,l_0)\mbox{ uniformly on }[\varepsilon,1-\varepsilon]$$
	for some small $\varepsilon>0$ with $[A,B]\subseteq \,]\varepsilon,1-\varepsilon[$. Using the definition of $\Gamma_n,\Gamma_0$ and \eqref{210911b} this implies
	\begin{equation}\label{210911d}
	\lim\limits_{n\to\infty}u_n(x)=u(x)\mbox{ uniformly on }[0,1].
	\end{equation}
	b) We now have to consider the first derivatives. We have by Theorem \ref{200911h}
	\begin{equation*}
	u_n'(x)=
	\begin{cases}
	\frac{\gamma}{1-\delta_n+\gamma x}, &x\leq a_n,\\
	g(x,x_{0,n},l_n), &a_n\leq x\leq b_n,\\
	-\frac{\gamma}{1-\delta_n-\gamma x}, &x\geq b_n,
	\end{cases}\mbox{  and define }\tilde{u}(x):=
	\begin{cases}
	\frac{\gamma}{1+\gamma x}, &x\leq A,\\
	g(x,x_{0},l_0), &A< x\leq B,\\
	-\frac{\gamma}{1-\gamma x}, &x> B.
	\end{cases}
	\end{equation*}
	By \eqref{210911b} and the continuity of $u_n'$ we get
	\begin{equation*}
	\frac{\gamma}{1+\gamma A}=\lim\limits_{n\to\infty}
	u_n'(a_n)=g(A,x_0,l_0)\mbox{ and }	-\frac{\gamma}{1-\gamma B}=\lim\limits_{n\to\infty}
	u_n'(b_n)=g(B,x_0,l_0). 
	\end{equation*}
	Hence $\tilde{u}$ is continuous on $[0,1]$ and it follows
	\begin{equation}\label{210911c}
	\lim\limits_{n\to\infty}u_n'(x)=\tilde{u}(x)\mbox{ uniformly on }[0,1].
	\end{equation}
	\eqref{210911d} together with \eqref{210911c} allows us to conclude $\tilde{u}=u'$.
	
	c) Here, we want to use the notation $g'(x,y,l)$ instead of $\frac{\partial}{\partial x}g(x,y,l)$. Since $u'=\tilde{u}$ is clearly differentiable in $[0,1]\setminus \{A,B\}$, we need to show the differentiability of $u'$ in $A,B$ and calculate $u''$. For the second derivatives we have
	\begin{equation*}
	u_n''(x)=
	\begin{cases}
	-\frac{\gamma^2}{(1-\delta_n+\gamma x)^2}, &x< a_n,\\
	g'(x,x_{0,n},l_n), &a_n< x< b_n,\\
	-\frac{\gamma^2}{(1-\delta_n-\gamma x)^2}, &x> b_n,
	\end{cases}
	\end{equation*}
	where we have discontinuities in $a_n,b_n$ due to (vii) and (viii) of Theorem \ref{200911h}. But at the same time (vii) and (viii) of Theorem \ref{200911h} imply
	$$-\frac{\gamma^2}{(1-\gamma B)^2}\leq
	\lim\limits_{n\to\infty}g'(b_n,x_{0,n},l_n)=g'(B,x_0,l_0)=
	\lim\limits_{n\to\infty}g'(\beta_n,x_{0,n},l_n)\leq
	-\frac{\gamma^2}{(1-\gamma B)^2}
	$$
	and
	$$-\frac{\gamma^2}{(1+\gamma A)^2}\leq
	\lim\limits_{n\to\infty}g'(a_n,x_{0,n},l_n)=g'(A,x_0,l_0)=
	\lim\limits_{n\to\infty}g'(\alpha_n,x_{0,n},l_n)\leq
	-\frac{\gamma^2}{(1+\gamma A)^2}.
	$$
	Therefore the function
	\begin{equation*}
	\tilde{\tilde{u}}(x):=
	\begin{cases}
	-\frac{\gamma^2}{(1+\gamma x)^2}, &x\leq A,\\
	g'(x,x_0,l_0), &A< x\leq B,\\
	-\frac{\gamma^2}{(1-\gamma x)^2}, &x> b_n,
	\end{cases}
	\end{equation*}
	is continuous. Furthermore, it holds $\tilde{\tilde{u}}(x)=u''(x)$ for every $x\in[0,1]\setminus \{A,B\}$. But since $u'$ has left-hand and right-hand derivatives in $A,B$, which are equal due to the continuity of $\tilde{\tilde{u}}$, $u'$ is differentiable on $[0,1]$ with derivative $\tilde{\tilde{u}}$ and therefore $u\in C^2([0,1],\R)$.
	
	For every $x\in[0,1]\setminus \{A,B\}$ almost all $n\in\N$ satisfy $x\notin\{a_n,b_n\}$ and we therefore have
	\begin{equation}\label{210911e}
	\lim\limits_{n\to\infty} u_n''(x)=u''(x).
	\end{equation}
	
	d) Now (iii) and (iv) of Theorem \ref{200911h} together with \eqref{210911c} and \eqref{210911e} imply $$Du(x)+f(x)-l\leq 0\mbox{ for all }x\in[0,1]$$ and $Du(x)+f(x)-l=0$ for all $x\in\,]A,B[$. $$Du(x)+f(x)-l=0\mbox{ for all }x\in [A,B]$$ then follows from the continuity of $u'$ and $u''$. It remains to show for all $x\in [0,1]$
	$$u'(x)\leq \frac{\gamma}{1+x\gamma},~u'(x)\geq -\frac{\gamma}{1-x\gamma}.$$
	We only show the latter, due to the analogous proof.
	
	e) Here, we show $u'(y)\geq -\frac{\gamma}{1-\gamma y}$ for all $y\in [0,1]$. By Theorem \ref{200911h} we have
	\begin{equation}\label{260911a}
	u(y)-u(x)+\Gamma_0(x,y)\leq 0 \mbox{ for all }x,y\in[0,1].
	\end{equation}
	Furthermore, we have $\frac{\partial}{\partial y}\Gamma_0(x,y)=\frac{\gamma}{1-\gamma y}$ for all $x> y\in\,]0,1[$. The assumption $$u'(y_0)<-\frac{\gamma}{1-\gamma y_0}\mbox{ for some }y_0\in\,]0,1[$$ implies $$u'(y)<\frac{-\gamma}{1-\gamma y}\mbox{ for all }y\in \,]y_0-\varepsilon,y_0+\varepsilon[\mbox{ for some suitable small }\varepsilon>0$$ and therefore
	\begin{equation}\label{260911b}
	u'(y)+\frac{\partial}{\partial y}\Gamma_0(y_0,y)<0\mbox{ for all }
	y\in \,]y_0-\varepsilon,y_0[.
	\end{equation}
	Now the function $h:[0,1]\rightarrow \R,y\mapsto u(y)-u(y_0)+\Gamma_0(y_0,y)$ satisfies
	\begin{equation}\label{260911c}
	h(y_0)=0\mbox{ and }h(y)\leq 0\mbox{ for all }y\in[0,1]
	\end{equation}
	due to \eqref{260911a}. But \eqref{260911b} implies that $h$ is strictly decreasing on $]y_0-\varepsilon,y_0]$, which contradicts \eqref{260911c}.
	
\end{proof}

Now we are able to prove a result, which was already announced in the last section. To the best of our knowledge there is no uniqueness result for the constants $A,B$ in the model from Section \ref{280911i}, not even for the similar models in \cite{TKA88} and \cite{AST01}.

\begin{satz}\label{280911h}
	Let the Merton fraction $\hat{h}:=\frac{\mu-r}{\sigma^2}$ satisfy $\hat{h}\in\,]0,1[$. Then there exist unique $A<B\in\,]0,1[$ such that the control limit policy for the limits $A,B$ is optimal.
\end{satz}

\begin{proof}
	Proposition \ref{200911e} guarantees the existence of some $A<B\in\,]0,1[$ and a function $u$ defined by \eqref{270911k} that satisfies the conditions of Theorem \ref{200911f} for the constants $A,B$. Hence the control limit policy $(L,M)$ for the limits $A,B$ is optimal.
	
	Now we take arbitrary $\tilde{A}<\tilde{B}\in\,]0,1[$ such that the control limit policy $(\tilde{L},\tilde{M})$ for the limits $\tilde{A},\tilde{B}$ is also optimal.
	
	a) We first show $\tilde{A}\leq A$ and $\tilde{B}\geq B$ but we will omit the proof for $\tilde{B}\geq B$ due to similarity.\\
	Let $l:=l_0$ and $g(x):=g(x,x_0,l_0)$ denote the constant and function introduced in the proof of Proposition \ref{200911e}, yielding $\rho=r+l$ and $u'(x)=g(x)$ on $[A,B]$. Let $f$ be as in \eqref{200911j} and let $D=x(1-x)\bigl(\mu-r-\sigma^2x\bigr)\frac{d}{dx}
	+\frac{1}{2}\sigma^2x^2(1-x)^2\frac{d^2}{dx^2}$ denote the generator from \eqref{eq:generator}. We have
	\begin{equation}\label{280911b}
	Du(x)=-f(x)+l\mbox{ on }[A,B],~g(A)=\frac{\gamma}{1+A\gamma},~
	g'(A)=-\frac{\gamma^2}{(1+A\gamma)^2}.
	\end{equation}
	What we need to show is
	\begin{equation}\label{280911c}
	u'(x)<\frac{\gamma}{1+\gamma x}\mbox{ for all }A<x\in\,]0,1[.
	\end{equation}
	We note that from Proposition \ref{200911e} we already know that $\leq$ holds for all $x\in\,[0,1].$
	and $u'(x)<0$ for all $x\geq B$, and so we have to consider $u'$ and hence $g$ on $[A,B]$. We further note that $\tilde{u}(x):=\int\limits_0^xg(y,x_0,l)dy$ is a classical solution to $D\tilde{u}(x)=-f(x)+l$ on $]0,1[$ (cf.$~$\cite{IS06}).
	Due to $u'(x)\leq\frac{\gamma}{1+\gamma x}$, every $x\in [A,B]$ with $u'(x)=\frac{\gamma}{1+\gamma x}$ satisfies $u''(x)=-\frac{\gamma^2}{(1+\gamma x)^2}$. Now \eqref{280911b} yields
	\begin{equation*}
	u''(x)=\frac{1}{\sigma^2x^2(1-x)^2}\bigl(
	2l-2(\mu-r)x+\sigma^2x^2-2x(1-x)\bigl(\mu-r-\sigma^2x\bigr)u'(x)
	\bigr).
	\end{equation*}
	To use the same argument as in \cite{IS06}, p. 932, we introduce the function
	\begin{equation*}
	v(x):=\frac{1}{\sigma^2x^2(1-x)^2}\Bigl(
	2l-2(\mu-r)x
	+\sigma^2x^2-2x(1-x)\bigl(\mu-r-\sigma^2x\bigr)\Bigl(\frac{\gamma}{1+\gamma x}\Bigr)
	\Bigr)+\frac{\gamma^2}{(1+\gamma x)^2},
	\end{equation*}
	which coincides with the derivative of $x\mapsto g(x)-\frac{\gamma}{1+\gamma x}$ in every $x\in [A,B]$ that satisfies $g(x)=\frac{\gamma}{1+\gamma x}$. We then calculate that
	$$v(x)=\frac{p(x)}{\sigma^2x^2(1-x)^2(1+\gamma x)^2},$$
	where $p$ is a polynomial of degree two, since the terms of degree three and four all cancel. Therefore $v(x)=0$ and hence $g(x)=\frac{\gamma}{1+\gamma x}$ has at most two solutions on $[A,B]$. But a maximum of $x\mapsto g(x)-\frac{\gamma}{1+\gamma x}$ on $[A,B]$ other than $A$ would imply at least three roots of its derivative on $[A,B]$ and hence three roots of $v$ and therefore \eqref{280911c} follows.
	Using $u$ we can calculate 
	\begin{align*}
	J_{v_0,\tilde{h}_0}\bigl(\tilde{L},\tilde{M}\bigr)
	&=r+l+\liminf\limits_{T\to\infty}\frac{1}{T}E\Biggl(
	\int\limits_0^T\Bigl(u'\bigl(\tilde{h}_s\bigr)\bigl(1+\tilde{h}_s\gamma\bigr)-\gamma\Bigr)
	\frac{1}{\tilde{V}_s}d \tilde{L}_s\\
	&~~~-\int\limits_0^T\Bigl(u'\bigl(\tilde{h}_s\bigr)\bigl(1-\tilde{h}_s\gamma\bigr)
	+\gamma\Bigr)	\frac{1}{\tilde{V}_s}d\tilde{M}_s
	+\int\limits_0^T\Bigl(Du\bigl(\tilde{h}_s\bigr)+f\bigl(\tilde{h}_s\bigr)-l\Bigr)ds\Biggr)\\
	&\leq r+l+\liminf\limits_{T\to\infty}\frac{1}{T}E\Biggl(
	\int\limits_0^T\Bigl(u'\bigl(\tilde{h}_s\bigr)\bigl(1+\tilde{h}_s\gamma\bigr)-\gamma\Bigr)
	\frac{1}{\tilde{V}_s}d \tilde{L}_s\Biggr)\\
	&=r+l+\Bigl(u'\bigl(\tilde{A}\bigr)\bigl(1+\gamma\tilde{A}\bigr)-\gamma\Bigr)
	\limsup\limits_{T\to\infty}\frac{1}{T}E\Biggl(
	\int\limits_0^T\frac{1}{\tilde{V}_s}d \tilde{L}_s\Biggr).
	\end{align*}
	Therefore $\tilde{A}>A$ would imply $J_{v_0,\tilde{h}_0}\bigl(\tilde{L},\tilde{M}\bigr)<r+l=\rho$ due to \eqref{280911c} and Lemma \ref{280911d}, which would pose a contradiction to the optimality of $(\tilde{L},\tilde{M})$.\leer
	
	b) It remains to prove $\tilde{A}\geq A$ and $\tilde{B}\leq B$. Here again, we restrict ourselves to showing $\tilde{A}\geq A$. Instead of $u$ we want to use the function
	\begin{equation}\label{280911e}
	\tilde{u}(x):=
	\begin{cases}
	u(A)+\int\limits_A^xg(y,x_0,l)dy, &0<x\leq B,\\
	u(B)+\Gamma_0(x,B), &x>B,
	\end{cases}
	\end{equation}
	on $]0,1[$, which is a classical solution to $D\tilde{u}(x)=-f(x)+l$ due to the proof of Proposition \ref{200911e}. Furthermore, since $g$ is a limit of functions $g_n$ that satisfy (ix) of Theorem \ref{200911h}, we have
	\begin{equation}\label{280911f}
	\tilde{u}'(x)=g(x,x_0,l)\leq\frac{\gamma}{1+\gamma x}\mbox{ for all }A\neq x\in\,]0,1[.
	\end{equation}
	Now the same argument as in a) applies here and so the inequality in \eqref{280911f} is strict.
	We use $\tilde{u}$ instead of $u$ to calculate as above 
	\begin{align*}
	J_{v_0,\tilde{h}_0}\bigl(\tilde{L},\tilde{M}\bigr)
	\leq r+l+\Bigl(\tilde{u}'\bigl(\tilde{A}\bigr)\bigl(1+\gamma\tilde{A}\bigr)-\gamma\Bigr)
	\limsup\limits_{T\to\infty}\frac{1}{T}E\Biggl(
	\int\limits_0^T\frac{1}{\tilde{V}_s}d \tilde{L}_s\Biggr),
	\end{align*}
	where we have used $\tilde{B}\geq B>A$ and hence $\tilde{u}'\bigl(\tilde{B}\bigr)\bigl(1-\tilde{B}\gamma\bigr)+\gamma=u'\bigl(\tilde{B}\bigr)\bigl(1-\tilde{B}\gamma\bigr)+\gamma\geq 0$. The strictness of the inequality \eqref{280911f} and Lemma \ref{280911d} therefore yield $\tilde{A}=A$.
\end{proof}

The following theorem is a main result of this section.

\begin{satz}\label{280911g}
	Let $(\delta_n)_{n\in\N}\in\,]0,1[^\N$ be a sequence of trading proportions defining the fixed costs in the sense of \eqref{cost} for a fixed common $\gamma>0$ defining the proportional costs and satisfying $\gamma<1-\sup\limits_{n\in\N}\delta_n$. Let the Merton fraction $\hat{h}:=\frac{\mu-r}{\sigma^2}$ satisfy $\hat{h}\in\,]0,1[$ and suppose $\lim\limits_{n\to\infty}\delta_n=0$. Then there exist constants $a_0<b_0\in\R$ and $\rho_0>0$ such that the corresponding constants $a_{n},\alpha_{n},\beta_{n}$ and $b_{n}$ given by Theorem \ref{210611a} for every $n\in\N$ together with the optimal growth rates $\rho_n$ satisfy
	\begin{equation*}
	\lim\limits_{n\to\infty}a_{n}=a_0
	=\lim\limits_{n\to\infty}\alpha_{n},~
	\lim\limits_{n\to\infty}\beta_{n}=b_0
	=\lim\limits_{n\to\infty}b_{n},~\lim\limits_{n\to\infty}\rho_n=\rho_0
	\end{equation*}
	and the control limit policy for the limits $A:=\varphi(a_0),B:=\varphi(b_0)$ is optimal for the optimization problem in the portfolio model with only proportional costs described in Section \ref{280911i} with the optimal growth rate $\rho_0$.
\end{satz}

\begin{proof}
	This is now a direct consequence of Proposition \ref{200911b} and \ref{200911e} together with Theorem \ref{280911h}.
\end{proof}


\begin{figure}[h]
	\centering
	\includegraphics[trim = 0mm 30mm 0mm 40mm,width=5cm]{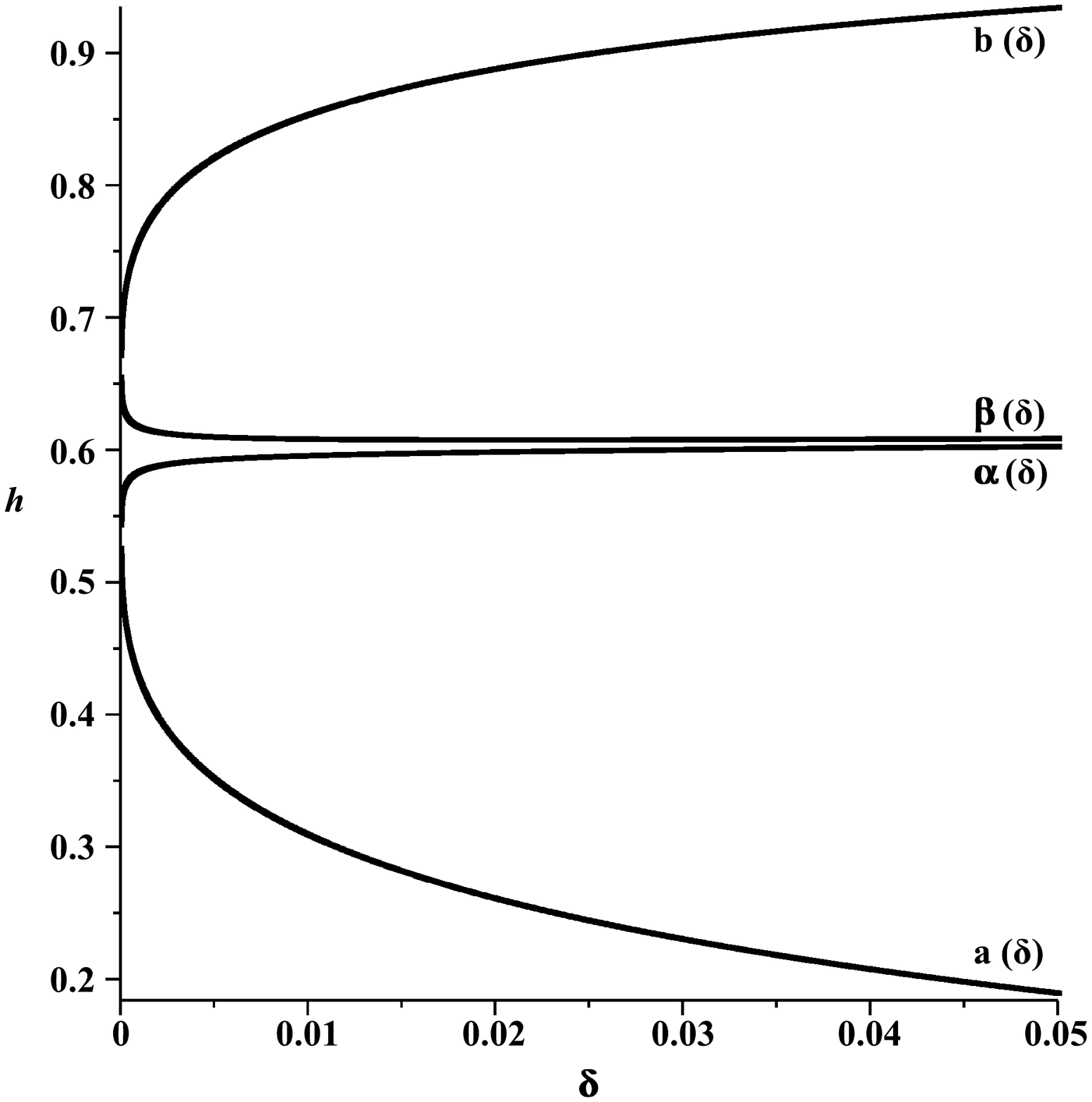}%
	\hspace{12mm}
	\includegraphics[trim = 0mm 30mm 0mm 40mm,width=5cm]{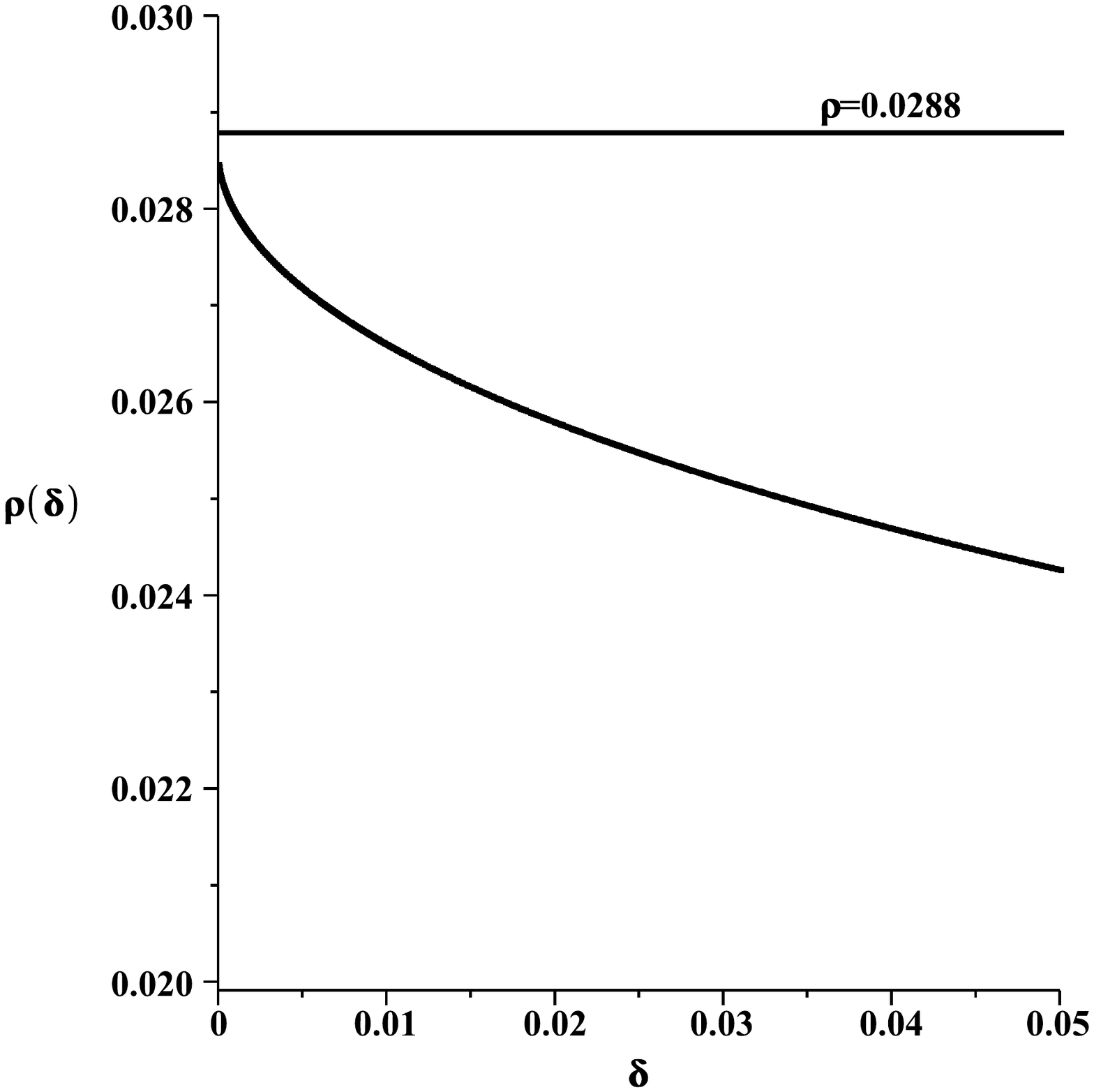}%
	\vspace{-4mm}
	\caption{Optimal boundaries and optimal growth rate as a function of $\delta$ for $r=0,~\sigma=0.4$,  $\mu=0.096$, and $\gamma=0.003$}%
	\label{pic:Diff10}%
\end{figure}


\subsection{Weak convergence of the optimal strategies}

We have proved so far the convergence of the boundaries $a_n,\alpha_n,\beta_n,b_n$ to the limits $a_0,b_0$ and the convergence of the corresponding optimal growth rates $\rho_n$ to $\rho_0$. Our aim in this subsection is the weak convergence of the corresponding risky fraction processes induced by the proportional constant boundary strategies $\widetilde{K}(a_n,\alpha_n,\beta_n,b_n),\,n\in\N$, to the risky fraction process induced by the control limit policy for the limits $a_0,b_0$, which is a diffusion with instantaneous reflection on $[a_0,b_0]$.

We begin with a characterization of the weak convergence via tightness and convergence in finite-dimensional distributions. For a brief review of the theory we refer 
to \cite{EK86} and \cite{JS03}. We will make use of the following

\begin{prop}\label{101011c}
	Let $(X^n)_{n\in\N_0}$ be a sequence of càdlàg processes defined on probability spaces $(\Omega_n,\mathcal{A}_n,P_n),\,n\in\N_0$. If $X^0$ is continuous, then we have for every dense subset $D\subseteq [0,\infty[$
	$$\Bigl(X^n\xrightarrow[]{~\mathcal{L}(D)~}X^0\mbox{ and } (X^n)_{n\in\N} \mbox{ tight}\Bigr)
	\Longleftrightarrow \Bigl(X^n\xrightarrow[]{~\mathcal{L}~}X^0\Bigr).$$
	Here, $\mathcal{L}(D)$ stands for convergence of the finite-dimensional distributions along $D$ and $\mathcal{L}$ for weak convergence. Here, a sequence $(X^n)_{n\in\N}$ converges {weakly} to $X^0$ if the laws $\mathcal{L}(X^n)$ converge weakly to $\mathcal{L}(X^0)$ in the set of probability measures on the Skorokhod space $D_{\R}[0,\infty)$.
\end{prop}


Since the limiting risky fraction process is a reflected diffusion on $[a_0,b_0]$, it is also a continuous process. Instead of proving the tightness of the sequence of our processes we will prove their $C$-tightness, which ensures tightness of the sequence and the continuity of the limiting process. We will not use this continuity but it is not a greater task to prove.

\begin{defi}
	A sequence $(X^n)_{n\in\N}$ of processes defined on probability spaces $(\Omega_n,\mathcal{A}_n,P_n),\,n\in\N$, is called $C$-\textit{tight} if it is tight and if for every probability measure $P$ on the Skorokhod space $D_\R[0,\infty[$ and every subsequence $(\mathcal{L}(X^{n_k}))_{k\in\N}$ that weakly converges to $P$, we necessarily have $P(C([0,\infty[,\R))=1$, where $C([0,\infty[,\R)$ denotes the set of continuous functions $f:[0,\infty)\rightarrow \R$.
\end{defi}

\begin{prop}\label{051011a}
	For a sequence of processes $(X^n)_{n\in\N}$ defined on probability spaces $(\Omega_n,\mathcal{A}_n,P_n)$, $n\in\N$, there is equivalence between
	\begin{enumerate}[{(}i{)}]
		\item $(X^n)_{n\in\N}$ is $C$-tight.
		\item For all $N\in\N$, $\varepsilon>0$ and $\eta>0$ there are $n_0\in\N$ and $\theta>0$ such that for all $n\geq n_0$
		\begin{equation*}
		P_n\Bigl(\sup\bigl\{|X^n_s-X^n_t|:s,t\in [0,N], |s-t|\leq \theta\bigr\}>\eta\Bigr)\leq \varepsilon.
		\end{equation*}
	\end{enumerate}
\end{prop}

%

Using this characterization, we can prove the $C$-tightness of our sequence of transformed risky fraction processes on $\R$. But since this is only a consequence of the convergence of the boundaries, we formulate it therefore independently from the optimality of the involved constant boundary strategies.

\begin{lem}\label{101011d}
	Let $a_n<\alpha_n\leq \beta_n<b_n,\,n\in\N_0,$ be a sequence of boundaries in $\R$ and suppose there exist $a_0<b_0\in\R$ satisfying
	\begin{equation}\label{121111a}
	\lim\limits_{n\to\infty}a_n=\lim\limits_{n\to\infty}\alpha_n=a_0<b_0=\lim\limits_{n\to\infty}\beta_n=\lim\limits_{n\to\infty}b_n.
	\end{equation}
	We denote by $Y^n$ the controlled diffusion corresponding to the constant boundary strategy $\overline{K}(a_n,\alpha_n,\beta_n,b_n)=$ $((\tau_k^n,\xi_{k}^n))_{k\in\N_0}$ 
	with $\xi_0^n=0$ and starting value $y_{0,n}\in \,]a_n,b_n[$ for every $n\in\N$. Then the sequence $(Y^n)_{n\in\N}$ is $C$-tight.
\end{lem}

\begin{proof}
	From the definition of a constant boundary strategy we have
	\begin{equation*}
	Y^n_t=y_{0,n}+\Bigl(\mu-r-\frac{\sigma^2}{2}\Bigr)\cdot t+\sigma W_t+\sum\limits_{k=0}^\infty \xi_k^n\mathbbm{1}_{\{\tau_k^n\leq t\}},\;\;\xi_0^n=0,\;\;\xi_{k}^n=
	\begin{cases}
	\alpha_n-a_n, &Y_{\tau_{k}^n}^{n}=a_n,\\
	\beta_n-b_n, &Y_{\tau_{k}^n}^{n}=b_n.
	\end{cases}
	\end{equation*}
	To prove (ii) of Proposition \ref{051011a} we take $N\in\N,\,\varepsilon>0,\,\eta>0$ and define $c:=\mu-r-\frac{\sigma^2}{2}$. From \eqref{121111a} we can assume without loss of generality $\inf\limits_{n\in\N}\beta_n-\alpha_n>0$ and take for simplicity $\eta<\inf\limits_{n\in\N}\frac{\beta_n-\alpha_n}{2}$. From the continuity of the paths $t\mapsto \sigma W_t+c t$ we can find some $\theta>0$ such that
	\begin{equation*}
	P\Bigl(\sup\bigl\{|\sigma(W_t-W_s)+c(t-s)|:s,t\in [0,N], |s-t|\leq \theta\bigr\}<\frac{\eta}{2}\Bigr)\geq 1-\varepsilon.
	\end{equation*}
	Furthermore, we can take by \eqref{121111a} some $n_0\in\N$ such that $\max\{b_n-\beta_n,\alpha_n-a_n\}\leq {\eta}/{2}.$ for all $n\geq n_0$.
	For every $\omega\in A:=\Bigl\{\omega\,:\,\sup\bigl\{|\sigma(W_t(\omega)-W_s(\omega))+c(t-s)|:s,t\in [0,N], |s-t|\leq \theta\bigr\}<\frac{\eta}{2}\Bigr\}$
	and every $n\geq n_0$ and $s,t\in [0,N]$ with $|s-t|\leq \theta$ we then have
	\begin{equation*}
	|Y_t^n(\omega)-Y_s^n(\omega)|\leq \max\{b_n-\beta_n,\alpha_n-a_n\}
	+|\sigma(W_t(\omega)-W_s(\omega))+c(t-s)|\leq \eta,
	\end{equation*}
	since there are only jumps in the same direction due to $\eta<\inf\limits_{n\in\N}\frac{\beta_n-\alpha_n}{2}$ and more than one jump then necessarily cancels some of the distance covered by the process $(\sigma W_t+c t)_{t\geq 0}$. Therefore it holds for all $n\geq n_0$
	\begin{equation*}
	P\Bigl(\sup\bigl\{|Y_t^n-Y_s^n|:s,t\in [0,N], |s-t|\leq \theta\bigr\}>\eta\Bigr)\leq 1-P(A)\leq\varepsilon.
	\end{equation*}
	
\end{proof}

Now we turn our attention to the convergence in finite-dimensional distributions. Here, it is again only a consequence of the convergence of the boundaries and no optimality is needed. What we actually will show is $\lim\limits_{n\to\infty} Y^n_t=Y_t$ for the following limit process $Y$.

We take $A<B\in\,]0,1[$ and the corresponding control limit policy $(L,M)\in\mathcal{A}_{v_0,h_0}$ for the limits $A,B$ from Definition \ref{200911g}. By 
It\^o's formula we have
\begin{equation}\label{051011b}
Y_t=y_0+\Bigl(\mu-r-\frac{\sigma^2}{2}\Bigr)\cdot t
+\sigma W_t+\frac{1+\gamma A}{A(1-A)}\int\limits_0^t
\frac{1}{V_s}d L_s
-\frac{1-\gamma B}{B(1-B)}\int\limits_0^t
\frac{1}{V_s}d M_s,
\end{equation}
where $Y=\psi(h)$, $\psi$ is the transformation function from Section \ref{reform} and $h$ is the corresponding risky fraction process. Now we define $c:=\mu-r-\frac{\sigma^2}{2}$ and
\begin{equation*}
Z_t^1:=\frac{1+\gamma A}{A(1-A)}\int\limits_0^t
\frac{1}{V_s}d L_s,~~Z_t^2:=
\frac{1-\gamma B}{B(1-B)}\int\limits_0^t
\frac{1}{V_s}d M_s,
\end{equation*}
and thus \eqref{051011b} becomes
\begin{equation}\label{051011d}
Y_t=y_0+c t+\sigma W_t+Z_t^1-Z_t^2.
\end{equation}
The processes $Z^1, Z^2$ are nondecreasing and we deduce for $a_0:=\psi(A)$ and $b_0:=\psi(B)$
\begin{equation}\label{051011c}
\int\limits_0^t\mathbbm{1}_{\{Y_s>a_0\}}dZ_s^1=\frac{1+\gamma A}{A(1-A)}\int\limits_0^t\mathbbm{1}_{\{h_s>A\}}\frac{1}{V_s}d L_s=0,
\end{equation}
since $\int\limits_0^t\mathbbm{1}_{\{h_s>A\}}dL_s=0$ by Definition \ref{200911g}, and analogously
\begin{equation}\label{131111b}
\int\limits_0^t\mathbbm{1}_{\{Y_s<b_0\}}dZ_s^2=0 \mbox{ and }Y_t\in[a_0,b_0]\mbox{
	for all }t\geq 0.
\end{equation}
Therefore we are in the framework to explicitly obtain 
the reflected processes as a solution to the so-called \textit{Skorokhod problem}, c.f.$~$Proposition 1.3 and Theorem 1.4 in \cite{KLRS07}.


\begin{lem}\label{051011e}
	In the situation of Lemma \ref{101011d} let $(Y^n)_{n\in\N}$ be a sequence of controlled processes with starting values $y_{0,n},\,n\in\N,$ and let $Y$ denote the process from \eqref{051011d} with starting value $y_0\in\,]a_0,b_0[$. We assume additionally $y_{0,n}=y_0\in \,]a_n,b_n[$ for every $n\in\N$. Then we have for almost all $\omega\in\Omega$
	\begin{equation*}
	\lim\limits_{n\to\infty}Y_t^n(\omega)=Y_t(\omega)\,\mbox{ for all }\,t\geq 0\mbox{ and hence } Y^n\xrightarrow[]{~\mathcal{L}(R_{\geq 0})~}Y.
	\end{equation*}
\end{lem}

\begin{proof}
	From the definition of a constant boundary strategy we have
	\begin{equation}
	\label{051011f}
	Y^n_t=y_0+c t+\sigma W_t+\sum\limits_{k=0}^\infty \xi_k^n\mathbbm{1}_{\{\tau_k^n\leq t\}},\;\;\xi_0^n=0,\;\;\xi_{k}^n=
	\begin{cases}
	\alpha_n-a_n, &Y_{\tau_{k}^n}^{n}=a_n,\\
	\beta_n-b_n, &Y_{\tau_{k}^n}^{n}=b_n.
	\end{cases}
	\end{equation}
	Now we can take a set $D$ of probability 1 such that the representations in \eqref{051011d}, and \eqref{051011f} hold pathwise for all $\omega\in D$ with continuity of the involved processes $W,Z^1,Z^2$ and all $n\in\N$.
	
	We define $\tau_0:=0$, $\tau_1:=\inf\bigl\{t>0:Y_t\notin \,]a_0,b_0[\bigr\}$ and inductively with $\inf\emptyset=\infty$ $$\tau_n:=\inf\bigl\{t>\tau_{n-1}:
	Y_t\in \{a_0,b_0\}\setminus\{Y_{\tau_{n-1}}\}\bigr\}\mbox{ on }\{\tau_{n-1}<\infty\}\mbox{ for all }n\geq 2.$$

	Comparing \eqref{051011f} and \eqref{051011d} it suffices to show for a fixed $\omega\in D$ and all $n\in\N_0$
	\begin{equation}\label{051011g}
	\lim\limits_{m\to\infty}\sum\limits_{k=1}^\infty \xi_k^m(\omega)\mathbbm{1}_{\{\tau_k^m(\omega)\leq t\}}=\bigl(Z_t^1-Z_t^2\bigr)(\omega)\mbox{ for all }t\in [\tau_n(\omega),
	\tau_{n+1}(\omega)[,\,\tau_n(\omega)<\infty.
	\end{equation}
	We define $Z_t:=Z_t^1-Z_t^2$ and prove \eqref{051011g} by induction on $n$ for this fixed $\omega\in D$ but to simplify matters we will suppress it in the following. We also note here that since the following arguments are pathwise, they are just about (deterministic) continuous functions and the (deterministic) jump-representation in \eqref{051011f}.
	
	a) For $n=0$ and $t<\tau_1$ we have $Y_s\in\,]a_0,b_0[$ and hence $Z^1_s=Z^2_s=0$ for all $s\leq t$ due to \eqref{051011c} and \eqref{131111b}. Furthermore, it holds $\lim\limits_{m\to\infty}a_m=a_0$ and $\lim\limits_{m\to\infty}b_m=b_0$ and we can therefore find some $m_0\in\N$ such that for all $m\geq m_0$
	$$y_0+cs+\sigma W_s\stackrel{\eqref{051011d}}{=}Y_s\in\,]a_m,b_m[\,\mbox{ for all }s\leq t\,\mbox{ and hence }\sum\limits_{k=1}^\infty \xi_k^m(\omega)\mathbbm{1}_{\{\tau_k^m(\omega)\leq t\}}\stackrel{\eqref{051011f}}{=}0=Z_t.$$
	
	b) Now we take $n\in\N$ and assume \eqref{051011g} to hold for $n-1$ and also $\tau_n<\infty$. Without loss of generality we further assume $Y_{\tau_n}=a_0$ and begin with $t=\tau_n$.
	
	b1) Using the convergence of the boundaries, we can find some $\tilde{\varepsilon}$ satisfying
	\begin{equation}\label{131111c}
	0<\tilde{\varepsilon}<\inf\limits_{m\geq \widetilde{m}_0}\frac{b_m-a_m}{4}\,\mbox{ for some }\widetilde{m}_0\in\N.
	\end{equation}
	Since the path $s\mapsto Y_s$ is continuous, we can then take $\tilde{u}<\tau_n$ such that $Y_s\in\, ]a_0,a_0+\tilde{\varepsilon}[$ for all $s\in ]\tilde{u},\tau_n[$. Let then $u:=\argmax\limits_{s\in [\tilde{u},\tau_n]}Y_s$, i.e.$~$we have
	$Y_s\in\,]a_0,Y_u]$ for all $s\in[u,\tau_n[$.
	Now we take $\varepsilon>0$ such that $\varepsilon<\frac{Y_u-a_0}{4}$ and also $\widetilde{m}_0\leq m_0\in\N$ such that
	\begin{equation}\label{131111d}
	|a_0-\alpha_m|,|\alpha_m-a_m|,\bigl|Y_u^m-Y_u\bigr|<\varepsilon\mbox{ for all }
	m\geq m_0.
	\end{equation}
	We fix $m\geq m_0$ and obtain from \eqref{131111c} and \eqref{131111d}
	\begin{equation}\label{131111e}
	\bigl|Y_u^m-\alpha_m\bigr|\leq\bigl|Y_u^m-Y_u\bigr|+|Y_u-a_0|+|a_0-\alpha_m|<\frac{b_m-a_m}{2}.
	\end{equation}
	By the definition of $u$, the path of $(c s+\sigma W_s)_{s\in[u,\tau_n]}$ covers a distance of at most $|Y_u-a_0|<\tilde{\varepsilon}$ in the upward direction starting at any time between $u$ and $\tau_n$.
	In the case that there is no jump of $Y^m$ in $a_m$ on $[u,\tau_n]$, we deduce $b_m-Y_u^m>|Y_u-a_0|$ from \eqref{131111e} and therefore $Y^m$ cannot reach $b_m$ on $[u,\tau_n]$. In the case of jumps in $a_m$ to $\alpha_m$ however, we obtain $b_m-\alpha_m>|Y_u-a_0|$ from \eqref{131111e} and $Y^m$ still cannot reach $b_m$ on $[u,\tau_n]$.
	Now we need to count the jumps of $Y^m$ on $[u,\tau_n]$.  Due to \eqref{131111d} it holds $$\bigl|Y_u^m-a_m\bigr|\geq |Y_u-a_0|-2\varepsilon.$$
	But starting from any new minimum point $s\in[u,\tau_n]$ of $Y$, e.g.$~$each time $Y^m$ reaches $a_m$, the path of $(c t+\sigma W_t)_{t\in[s,\tau_n]}$ covers a distance of at most $|Y_u-a_0|-|Y_u-Y_s|$ in the downward direction on $[s,\tau_n]$ by the definition of $u$. Therefore, after the first jump in $a_m$, we only have a distance of at most $2\varepsilon$ left for jumping to $\alpha_m$ and then returning to $a_m$ again. This way we obtain for the number of possible jumps at $a_m$ of $Y^m$
	\begin{equation}\label{131111f}
	\bigl|\bigl\{k:u<\tau_k^m\leq\tau_n\bigr\}\bigr|\in
	\left[0,\frac{2\varepsilon}{\alpha_m-a_m}+1\right].
	\end{equation}
	Finally, we have
	\begin{equation*}
	Z_u=Z_{\tau_n}\mbox{ and }
	\sum\limits_{k=1}^\infty \xi_k^m
	\mathbbm{1}_{\{\tau_k^m\leq \tau_n\}}
	=\sum\limits_{k=1}^\infty \xi_k^m
	\mathbbm{1}_{\{\tau_k^m\leq u\}}
	+\sum\limits_{k=1}^\infty (\alpha_m-a_m)
	\mathbbm{1}_{\{u<\tau_k^m\leq \tau_n\}}
	\end{equation*} and we therefore get from \eqref{131111f}
	\begin{align*}
	\biggl|\sum\limits_{k=1}^\infty \xi_k^m
	\mathbbm{1}_{\{\tau_k^m\leq \tau_n\}}-Z_{\tau_n}\biggr|
	&\leq\biggl|\sum\limits_{k=1}^\infty \xi_k^m
	\mathbbm{1}_{\{\tau_k^m\leq u\}}-Z_u\biggr|+\biggl|\sum\limits_{k=1}^\infty (\alpha_m-a_m)
	\mathbbm{1}_{\{u<\tau_k^m\leq \tau_n\}}\biggr|\\
	&\leq\left|Y_u^m-Y_u\right|+2\varepsilon+\alpha_m-a_m\leq 4\varepsilon.
	\end{align*}
	
	b2) Now let $t\in \,]\tau_n,\tau_{n+1}[$. We note that our path satisfies $Z_s^2=Z_{\tau_n}^2$ for all $s\in[\tau_n,\tau_{n+1}]$ due to $Y_{\tau_n}=a_0$. By \eqref{051011c} and \eqref{131111b} we are in the framework of instantaneous reflection, see above, and obtain therefore the explicit representation for the path of the process $Z$, whose growth from $\tau_n$ onward is given via
	\begin{equation}\label{061011b}
	Z_t-Z_{\tau_n}=-\inf\limits_{s\in[\tau_n,t]}\bigl(y_0+c s+\sigma W_s+Z_{\tau_n}-a_0\bigr).
	\end{equation}
	Let $\varepsilon>0$ such that $\varepsilon<\frac{b_0-a_0}{4}$. Since $Y_{\tau_n}=a_0$ and $Y$ does not reach $b_0$ on $]\tau_n,\tau_{n+1}[$ we can find by b1) and the convergence of the boundaries some $m_0\in\N$ such that for all $m\geq m_0$ the path of $Y^m$ does not reach $b_m$ on $]\tau_n,\tau_{n+1}[$ and therefore has no jumps at $b_m$ and additionally satisfies
	\begin{equation}\label{061011c}
	\biggl|\sum\limits_{k=1}^\infty \xi_k^m
	\mathbbm{1}_{\{\tau_k^m\leq \tau_n\}}-Z_{\tau_n}\biggr|=\bigl|Y_{\tau_n}^m-Y_{\tau_n}\bigr|<\varepsilon\,\mbox{ and }\bigl|Y_{\tau_n}^m-a_m\bigr|,|\alpha_m-a_m|
	<\varepsilon.
	\end{equation}
	Between $\tau_n$ and $t$ the path of $(c s+\sigma W_s)_{s\in[\tau_n,t]}$ covers a distance of $\bigl|Z_t-Z_{\tau_{n+1}}\bigr|$ in the downward direction by \eqref{061011b} and from any new minimum point $s\in[t,\tau_{n+1}[$ of $Y$, e.g.$~$each time $Y^m$ reaches $a_m$, also $\bigl|Z_s-Z_{\tau_{n+1}}\bigr|$. Hence, we can calculate the number of possible jumps of $Y^m$ at $a_m$ by the same method as in b1) and get due to $\bigl|Y_{\tau_n}^m-a_m\bigr|<\varepsilon$
	\begin{equation*}
	\bigl|\bigl\{k:\tau_n<\tau_k^m\leq t\bigr\}\bigr|\in
	\left[\frac{Z_t-Z_{\tau_n}-\varepsilon}{\alpha_m-a_m}
	,\frac{Z_t-Z_{\tau_n}}{\alpha_m-a_m}+1 \right].
	\end{equation*}
	Therefore \eqref{061011c} then implies for all $m\geq m_0$
	\begin{align*}
	\biggl|\sum\limits_{k=1}^\infty \xi_k^m
	\mathbbm{1}_{\{\tau_k^m\leq t\}}-Z_t\biggr|
	&\leq\biggl|\sum\limits_{k=1}^\infty \xi_k^m
	\mathbbm{1}_{\{\tau_k^m\leq \tau_n\}}-Z_{\tau_n}\biggr|+\biggl|\sum\limits_{k=1}^\infty (\alpha_m-a_m)
	\mathbbm{1}_{\{\tau_n<\tau_k^m\leq t\}}+Z_{\tau_n}-Z_t\biggr|\\
	&\leq\left|Y_{\tau_n}^m-Y_{\tau_n}\right|+\varepsilon\leq 2\varepsilon.
	\end{align*}
\end{proof}

Now we are able to prove one of the main theorems of this section. It applies especially for the transformed risky fraction processes in the situation of Theorem \ref{280911g}.
\begin{satz}\label{101011i}
	In the situation of Lemma \ref{051011e} we have
	\begin{equation*}
	Y^n\xrightarrow[]{~\mathcal{L}~}Y.
	\end{equation*}
\end{satz}

\begin{proof}
	The assertion follows directly from Proposition \ref{101011c}, Lemma \ref{101011d} and Lemma \ref{051011e}.
	
\end{proof}

The same result is also true for the risky fraction processes on $[0,1]$. We use the transformation function $\varphi=\psi^{-1}:\R\rightarrow \,]0,1[$ from Section \ref{reform}.

\begin{satz}\label{101011j}
	In the situation of Lemma \ref{051011e} we have
	\begin{equation*}
	h^n\xrightarrow[]{~\mathcal{L}~}h
	\end{equation*}
	for the corresponding risky fraction processes $h_t^n=\varphi(Y_t^n)$, $n\in\N$, and $h_t=\varphi(Y_t)$.
\end{satz}

\begin{proof}
	Since $\varphi$ is Lipschitz continuous, and thus Lemmata \ref{101011d} and \ref{051011e} are also true for $(h^n)_{n\in\N}$ and $h$. Proposition \ref{101011c} then yields the assertion.
	
\end{proof}

\bibliographystyle{apt}
\bibliography{bibs/literatur}

\end{document}